\documentclass[a4paper]{article}

\usepackage[english]{babel}
\usepackage[utf8]{inputenc}
\usepackage{amsmath}
\usepackage{amsthm}
\usepackage{amsfonts}
\usepackage[breakable,skins]{tcolorbox}
\usepackage{tabularx}
\usepackage{booktabs}
\usepackage{dsfont}
\usepackage{hyperref}
\usepackage[top=.75in,bottom=.75in,left=1in,right=1in]{geometry}
\usepackage{authblk}

\definecolor{b1}{RGB}{0,101,189} 
\definecolor{a1}{RGB}{227,114,34} 

\newcommand{\ket}[1]{|#1\rangle} 
\newcommand{\Hil}{\mathcal{H}} 
\newcommand{\id}{\mathds{1}}

\newcommand{\dom}{\operatorname{dom}}

\newcommand{\Lie}[1]{\left\langle #1 \right\rangle_{\text{Lie}} } 
\newcommand{\LieC}[1]{\left\langle #1 \right\rangle_{\text{Lie},\mathbb{C}} } 
\newcommand{\SU}{\widetilde{\mathcal{SU}}}
\newcommand{\U}{\widetilde{\mathcal{U}} }
\newcommand{\su}{\widetilde{\mathfrak{su}}}
\newcommand{\G}{\mathcal{G}} 
\newcommand{\lC}{\mathfrak{l}_\mathbb{C}} 
\newcommand{\lCn}{\mathfrak{l}^{(n)}_\mathbb{C}}  
\newcommand{\lCK}{\mathfrak{l}^{[K]}_\mathbb{C}}  

\newcommand{\HP}{H_{\text{\tiny{P}}} } 
\newcommand{\HA}{H_{\text{\tiny{A}}} } 
\newcommand{\HI}{H_{\text{\tiny{I}} }} 
\newcommand{\HH}[1]{H_{\mathsf{\tiny{H}}}(#1) }
\newcommand{\HHtwo}{H_{\text{\tiny{H}}} }
\newcommand{\oPk}{\omega^{\text{\tiny{P}}}_k }
\newcommand{\oAk}{\omega^{\text{\tiny{A}}}_k }
\newcommand{\oIk}{\omega^{\text{\tiny{I}}}_k }
\newcommand{\oP}[1]{\omega^{\text{\tiny{P}}}_{#1} }
\newcommand{\oA}[1]{\omega^{\text{\tiny{A}}}_{#1} }
\newcommand{\oI}[1]{\omega^{\text{\tiny{I}}}_{#1} }

\newcommand{\tH}{{\textsf{\tiny{H}}} }

\newcommand{\DriftJ}{H_\mathsf{\tiny{drift}}} 

\newcommand{\JC}{\textsf{\tiny{JC}} } 
\newcommand{\JCH}{\mathsf{\tiny{JCH}} } %
\newcommand{\fset}{D_M} 

\newtheorem{dfn}{Definition}[section]
\newtheorem{prop}[dfn]{Proposition}
\newtheorem{thm}[dfn]{Theorem}
\newtheorem{lem}[dfn]{Lemma}

\title{Controllability of the Jaynes-Cummings-Hubbard model}
\author[1]{Margret Heinze\footnote{margret.heinze@ma.tum.de}}
\author[1,2]{Michael Keyl\footnote{michael.keyl@tum.de}}
\affil[1]{Zentrum Mathematik, Technische Universit\"at M\"unchen, 85748 Garching, Germany}
\affil[2]{Dahlem Center for Complex Quantum Systems, Freie Universit\"at Berlin, 14195 Berlin, Germany}

\date{}
\vspace{-8ex}
\begin{document}
	
\maketitle

\begin{abstract}

In quantum control theory, the fundamental issue of controllability covers the questions whether and under which conditions a system can be steered from one pure state into another by suitably tuned time evolution operators. 
Even though Lie theoretic methods to analyze these aspects are well-established for finite dimensional systems, they fail to apply to those with an infinite number of levels. The Jaynes-Cummings-Hubbard model -- describing two-level systems in coupled cavities -- is such an infinite dimensional system.

In this contribution we study its controllability. In the two cavity case we exploit symmetry arguments; we show that one part of the control Hamiltonians can be studied in terms of infinite dimensional block diagonal Lie algebras while the other part breaks this symmetry to achieve controllability. An induction on the number of cavities extends this result to the general case. Individual control of the qubit and collective control of the hopping between cavities is sufficient for both pure state and strong operator controllability. We additionally establish new criteria for the controllability of infinite dimensional quantum systems admitting symmetries.

\end{abstract}

\section{Introduction}

The experimental investigation of a quantum system demands control over its dynamics. In this context, the interdisciplinary and rapidly evolving field of quantum control theory aims at establishing a theoretical footing and systematic methods. The implicit goal is to be able to drive systems into a desired target state by time-dependently tuning selected couplings. Reviews can be found in~\cite{dong2010quantum,glaser2015training}.
One of the fundamental questions in quantum control theory concerns a system's controllability. It asks whether and under which conditions a system can be steered into any desired state. 
There exist other variants of controllability, also asking whether it is possible to implement any unitary operator on the system by a suitably tuned time evolution operator. 

This notion connects controllability to quantum computation where certain universal unitary gate sets have to be implemented. Promising proposals for the physical implementation of a quantum computer use a trapped ion as basis~\cite{cirac1995quantum}, a single atom in a cavity~\cite{ye1999trapping}, a superconducting qubit in a microwave resonator~\cite{shnirman1997quantum} or a quantum dot qubit in a nanocavity~\cite{imamog1999quantum}. They range from quantum optical systems to superconductor or semiconductor solid state systems.
All such systems have in common that they involve interactions of light and matter and are described by cavity quantum electrodynamics.
The light fields contain even in the simplest description an infinite number of levels. 
Many of the exciting and challenging control problems concern infinite-dimensional systems that treat light-matter interaction.

Whereas the controllability of finite dimensional systems is well-studied and treatable via the Lie algebra rank condition \cite{sussmann1972controllability, jurdjevic1972control, brockett1972system, brockett1973lie}, it is more intricate in infinite dimensions since one has to deal with technical difficulties such as unbounded operators, different norms and the potential irreversibility of the dynamics. 
For infinite dimensional systems, Huang et al.~\cite{huang1983controllability} laid foundations introducing the concept of analytic controllability but their generic result is a no-go theorem. After further negative results concerning the exact controllability~\cite{ball1982controllability, turinici2012mathematical} different notions of approximate controllability were developed. Several approaches to handle the above mentioned problems have been introduced and used to study controllability of different relevant systems~\cite{adami2005controllability,beauchard2005local,wu2006smooth,nersesyan2009growth, chambrion2009controllability,beauchard2010local, bliss2014quantum,keyl2014controlling, boscain2015approximate}. 
However, for infinite dimensional systems, no general controllability criteria exist and only few universal results are known. Hence it is necessary to develop new tools and analyze the controllability of relevant examples.

One natural candidate is the famous Jaynes-Cummings (JC) model~\cite{jaynes1963comparison} that is of wide experimental and theoretical interest in the fields of quantum optics and solid state physics. Describing a two-level system (qubit) that interacts with one mode of a quantum harmonic oscillator (cavity mode) the JC model forms the basis of cavity quantum electrodynamics and is used to model the above mentioned proposals for a quantum computer. 
The Jaynes-Cummings-Hubbard (JCH) model, also referred to as Jaynes-Cummings lattice, generalizes the JC model in the following sense: it describes an array or lattice of JC models that couple to each other via a Hubbard like interaction, i.e.~hopping of bosonic excitations.
It was introduced in 2006 and 2007 independently in~\cite{angelakis2007photon,hartmann2006strongly,greentree2006quantum} to describe quantum phase transitions. In circuit quantum electrodynamics it is used as a basic model when considering circuits of superconducting qubits~\cite{schmidt2013circuit} and has been experimentally studied (two cavity JC dimer in a superconducting circuit platform~\cite{raftery2014observation}).
However, its controllability has not been discussed yet.

This article studies its controllability in the full infinite dimensional setting using approximations in the strong operator topology. We treat both notions, approximating pure states and unitary operators. 
Our main contribution is the following statement: For the Jaynes-Cummings-Hubbard model, individual control of all the qubits and collective control of the hopping interaction between cavities is sufficient for the model's controllability. Hence we can find control functions to approximately tune the system from a given initial pure state to any desired target state (pure-state controllability); we can also approximate any unitary operator on the system in the strong operator topology by a suitably tuned time evolution operator (operator controllability).

This article is organized as follows: In Section~\ref{sec:Problem_Results}, we define controllability and state our main result, a JCH control system that is controllable. The proof builds on former results by Keyl et al.~\cite{keyl2014controlling} who introduced a strategy to analyze an infinite dimensional system's controllability by exploiting its symmetries. In Section~\ref{sec:symmetry}, we generalize this strategy to systems involving a non-tunable drift Hamiltonian and give a list of sufficient conditions for a system's controllability. 
The remainder of the article considers the JCH model. 
In Section~\ref{sec:spectral_self}, we use spectral analysis for unbounded operators to show that the JCH model satisfies the requirements of a quantum control system and analyze self-adjointness and recurrence of the control and drift Hamiltonians. 
Section~\ref{sec:2cavities} considers the JCH model for two cavities where we check the sufficient conditions for controllability from Section~\ref{sec:symmetry}. Section~\ref{sec:graph} extends this result to the general setting with $M\ge2$ cavities by treating the JCH model as a graph and exploiting a result on the controllability of overlapping systems. We also include a detailed discussion of our results. In Section~\ref{sec:conclusion}, we conclude with an outlook.

\section{Description of the problem and main result}\label{sec:Problem_Results}

\subsection{Quantum control theory in infinite dimensions}

The dynamics of a closed quantum system are described by the Schrödinger equation 
\begin{equation}\label{Schro}
	i\frac{d}{dt} \ket{\psi (t)}= H(t)\, \ket{\psi(t)}
\end{equation} with initial condition $ \ket{\psi(0)}= \ket{\psi_0} $. Here the system's pure state at time $t\in\mathbb{R}_+$ is described by a vector $ \ket{\psi(t)}\in \Hil$, i.e.~an element of a separable, potentially infinite dimensional Hilbert space $\Hil$. In quantum control theory, the system's Hamiltonian $H(t)$ can be written as
\begin{equation}\label{ConHam}
	 H(t)= H_0 + \sum_{k=1}^{d} u_k(t) H_k \; .
\end{equation}The drift Hamiltonian $H_0$ describes the system's free and uncontrolled evolution. The system is coupled to external controls via the control Hamiltonians $H_1$, $\dots,$ $H_d$. The amplitudes of these couplings can be time-dependently tuned what is represented by the scalar control functions $u_1$, $\dots,$ $u_d$.

\begin{dfn}[Quantum control system] \label{dfn:con_sys}
	Consider a quantum system such that its dynamics are described by Eq.~\eqref{Schro} with Hamiltonian of the form of Eq.~\eqref{ConHam}. A quantum control system is characterized by a tuple $(H_0,H_1,\dots,H_d,\mathcal{H}, \mathcal{A})$. The set $\mathcal{A}$ is a subset of the real-valued functions and denotes the set of admissible control functions. \\
	In this contribution, we make the following additional assumptions:
	\begin{itemize}
		\item[(i)] The set of admissible controls is $\hat{\mathcal{A}}:=\big\{(u_1,\dots,u_d)|\, u_1,\dots, u_d:\mathbb{R}_+ \to \mathbb{R},\text{ piecewise}\\ \text{constant, only one of them non-zero at each time } t\big\}$; \label{it:admissible}
		\item[(ii)] $H_0$ and $H_0+ H_k$ are self-adjoint for all $k=1,\dots, d$;
		\item[(iii)] the identity operator $\id$ is among the control Hamiltonians.
	\end{itemize}
\end{dfn}

Note that we do not restrict to finite dimensional Hilbert spaces. This has
several consequences. First, we have to specify a topology when discussing convergence and related issues. This will be the strong operator topology. Second, operators on $\Hil$ might be unbounded. In fact, for many interesting physical systems, operators such as position and momentum or bosonic creation and annihilation operators are unbounded.
The resulting (domain) problems when adding unbounded operators or taking exponentials of them motivate the assumptions in Definition~\ref{dfn:con_sys} as follows: 
$(i)$ and $(ii)$ imply that the time evolution operator $U(t)$ that is given by the relation $\ket{\psi(t)}=U(t)\ket{\psi(0)}$, is well-defined. This can be easily seen in the following argument: Due to~$(i)$, the Hamiltonian is piecewise constant and $U(t)$ is hence given by 
\begin{equation}\label{eq:time_evolution_operator}
	U(t)=e^{ i\Delta t_1(H_0+ \hat{u}_1 H_{k_1}) } \dots e^{ i\Delta t_n (H_0+ \hat{u}_n  H_{k_n})}
\end{equation} where $\Delta t_1+\dots+\Delta t_n=t$, $\Delta t_i\ge 0$, $\hat{u}_i\in\mathbb{R}$ and $k_1,\dots, k_n\in \{1,\dots, d\}$ for some $n\in\mathbb{N}$. As a finite product of exponentials of anti-self-adjoint operators -- by~$(ii)$ -- the time evolution operator~\eqref{eq:time_evolution_operator} is well-defined. 
We impose condition $(iii)$ for technical reasons; to ignore global phases which do not have any physical meaning.

Via tuning the control functions, one steers the system's Hamiltonian $H(t)$ and hence also $U(t)$. We consider the question which states can be reached by such a tuned time evolution operator from a given initial state.

\begin{dfn}[Pure-state controllability] 
	A quantum control system $(H_0,H_1,\dots,H_d,\mathcal{H}, \hat{\mathcal{A}})$ is pure-state controllable if and only if for all $\epsilon>0$ and $\ket{\psi_1},\ket{\psi_2}$ in the complex unit sphere in $\Hil$ there exist admissible controls $(u_1,\,\dots, \, u_d)\in \hat{\mathcal{A}}$ and a finite time $t\ge0$ such that $\|U(t)\ket{\psi_1}-\ket{\psi_2}\|\le \epsilon$.
\end{dfn}

This analysis can be lifted to the level of operators where the Schrödinger equation becomes $ i\frac{d}{dt} U(t)= H(t) U(t)$ with initial condition $U(0)=\id$. Let $\mathcal{U}(\Hil)$ denote the group of unitaries on $\Hil$. We clearly know that all possible time evolution operators are unitary but we can ask if the converse is also true. 

\begin{dfn}[Strong operator controllability] 
	The set of reachable unitaries $\mathcal{R}$ of a quantum control system $(H_0,H_1,\dots,H_d,\mathcal{H}, \hat{\mathcal{A}})$ is defined as the set of unitaries $U$ on $\Hil$ that satisfy:
	for each $ \epsilon>0$, and every finite number of pure states $\ket{\psi_1},\dots,\ket{\psi_f}\in \Hil$, where $f\in \mathbb{N}$, there exist admissible controls $(u_1,\,\dots, \, u_d)\in \hat{\mathcal{A}}$ and a finite time $t\ge0$ such that $ \|\big(U-U(t)\big)\ket{\psi_k} \|\le\epsilon$ for all $k=1,\dots,f $. 
	The quantum control system is strongly operator controllable if and only if $\mathcal{R}= \mathcal{U}(\Hil)$. 
\end{dfn}

In systems satisfying this condition, every unitary operator can be approximated in the strong operator topology by a suitably tuned time evolution operator. Clearly, strong operator controllability implies pure-state controllability. 
Note that in both definitions of controllability, we do not care about global phases since the identity operator is by assumption among the control Hamiltonians.

\subsection{Controllability of the Jaynes-Cummings-Hubbard model}

The Jaynes-Cummings-Hubbard (JCH) model describes a finite number of coupled cavities, each of which contains a quantum harmonic oscillator mode and a two-level system. Although the two-level system may be realized in different physical systems (electron in a quantum dot, superconducting qubit, $\dots$) we will refer to it as an atom. An atomic excitation of zero (one) depicts the atom's ground (excited) state whereas the harmonic oscillator excitations describe photons. For $M$ cavities, the system's Hilbert space is given by the $M$-fold tensor product of the photon Hilbert space $\text{L}^2(\mathbb{R})$ of square integrable functions and the atomic Hilbert space $\mathbb{C}^2$ so that
\begin{equation}
	\Hil_M:=\left(\text{L}^2 (\mathbb{R} ) \otimes \mathbb{C}^2 \right)^{\otimes M}.
\end{equation}
A basis of $\Hil_M$ is hence given by the tensor product of the canonical basis $b_i$ of $\mathbb{C}^2$ and the Hermite polynomials in number basis $m_i$ of $\text{L}^2(\mathbb{R})$. Let
\begin{equation}\label{eq:D}
	\fset :=\text{span} \{ \ket{m_1}\otimes \ket{b_1}\otimes \cdots \otimes \ket{m_M}\otimes \ket{b_M} |b_i\in\{0,1\},\, m_i\in \mathbb{N}_0,\, i=1,\dots, M\}
\end{equation} 
be the finite excitation subspace, that is spanned by vectors with finite photonic and atomic excitation number. 
On this set, we define the Jaynes-Cummings-Hubbard Hamiltonian  
\begin{equation}
	H_{\JCH}= \sum_{j=1}^M \left[ \oP{i} a^*_ia^{\phantom{*}}_i +\oA{i} \sigma^z_i+ \oI{i} (a^*_i \sigma^-_i +a^{\phantom{*}}_i \sigma^+_i) \right] + \sum_{i\sim j\in I} \omega^\tH_{i,j} (a^*_i a^{\phantom{*}}_j  +a^{\phantom{*}}_i a^*_j)\; .
\end{equation}
The operators $a^*$ and $a$ denote the bosonic creation and annihilation operators, respectively, and $\sigma^x$, $\sigma^y$ and $\sigma^z$ are the three Pauli matrices where we set $\sigma^\pm=(\sigma^x\pm i\sigma^y)/2$. Let $\oA{i}$, $\oP{i}$, and $\oI{i}$ be families of real values, where in cavity $i$ $\oA{i}$ denotes the atom's transition energy, $\oP{i}$ the energy cost for creating a photon, and $\oI{i}$ the strength of the atom-photon coupling. 
Lower indices $i$ and $j$ label the cavity that the corresponding operator acts on, so that the notation $\sigma^x_i$ is for example short hand for $=\id^{\otimes 2(i-1)} \otimes (\id \otimes \sigma^x) \otimes \id^{\otimes2(M-i)} $. 
The coupling between different cavities is given by the photon hopping term $\sum_{i\sim j\in I} \omega^\tH_{i,j} (a^*_i a^{\phantom{*}}_j  +a^{\phantom{*}}_i a^*_j)$ where~$\omega^\tH_{i,j}$ denotes the hopping rate between cavities $i$ and $j$. Hence the set $I$ represents the hopping structure in the following sense: a tuple $(i,j)$ is in $I$ if and only if hopping between cavities $i$ and $j$ is possible. We call $I$ the system's hopping interaction graph that is illustrated in Fig.~\ref{fig:JCHmodel}(A).

A single cavity $i$ is described by the JC model with Hamiltonian
\begin{equation}\label{eq:HJC_i}
	H_{\JC,i}= \oP{i} a^*_ia^{\phantom{*}}_i +\oA{i} \sigma^z_i+ \oI{i} (a^*_i \sigma^-_i + a^{\phantom{*}}_i \sigma^+_i )\, .
\end{equation} It contains an atom and a harmonic oscillator mode that interact via $a^*_i \sigma^-_i + a^{\phantom{*}}_i \sigma^+_i$, i.e.,~via exchange of excitations. This models spontaneous emission and absorption of photons. 
The controllability of the JC model has been studied, first restricted to a finite dimensional truncation, where the notions of approximating pure states~\cite{rangan2004control} and unitaries~\cite{yuan2007controllability} were considered. These articles were followed by the study of related models~\cite{ervedoza2009approximate, boscain2015control,paduro2015approximate} and an extension to the full infinite dimensional setting for pure states~\cite{brockett2003controllability,bloch2010finite,boscain2015control} and unitaries~\cite{keyl2014controlling}. In the latter, the authors studied convergence in the strong operator topology, using a symmetry strategy.
In this contribution, we do the same for the JCH model. So let us first define the control system associated with the JCH model.

\begin{figure}[ht]
	\centering
	\def\svgwidth{14cm}
	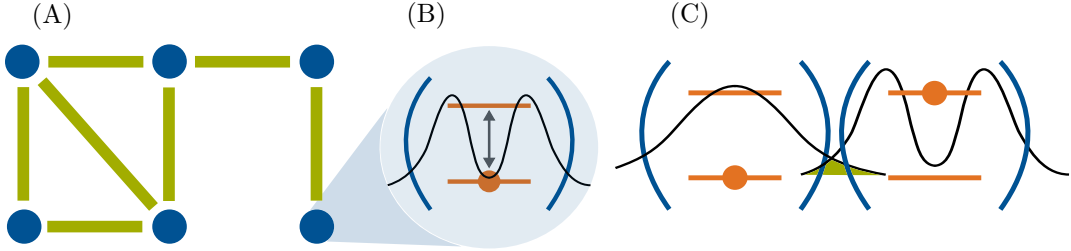
	\caption{(A) Graphical representation of a JCH model with six cavities (blue dots) where photon hopping between two cavities is represented by the (green) edges of a graph, called the hopping interaction graph. (B) One cavity -- depicted by two blue opposing mirrors -- is described by the Jaynes-Cummings model; it contains a two-level system (orange) and one harmonic oscillator mode (black). (C) JCH model for two cavities: Due to overlapping harmonic oscillator modes between neighboring cavities (green filling) photons can hop from one cavity to another.}
	\label{fig:JCHmodel}
\end{figure}

\begin{prop}[JCH Control System] \label{prop:JCH_con_sys}
	Consider the Hilbert space $\Hil_M=\left(\text{L}^2 (\mathbb{R} ) \otimes \mathbb{C}^2 \right)^{\otimes M}$ for some integer $M\ge 2$. On the set $\fset$ from Eq.~\eqref{eq:D}, define 
	\begin{align}\label{eq:drift_JCH}
		\DriftJ &:= \sum_{i=1}^M \left[ \oP{i} a^*_ia^{\phantom{*}}_i +\oA{i} \sigma^z_i+ \oI{i} (a^*_i \sigma^-_i + a^{\phantom{*}}_i \sigma^+_i )\right]
	\end{align} where $\{\oA{i}\}_{i=1,\dots,M},\{\oP{i}\}_{i=1,\dots,M},\{\oI{i}\}_{i=1,\dots,M}$ denote families of real, positive constants. Let $I$ be a subset of $\{(i,j)|1\le i<j\le M\}$ such that its elements $(i,j)\in I$ form edges of a connected graph. 
	Set 
	\begin{equation}\label{eq:def_HHI}
		\HH{I}:= \sum_{j\sim k\in I} (a^*_j a^{\phantom{*}}_k  +a^{\phantom{*}}_j a^*_k)\; .
	\end{equation}
	Then $(\DriftJ,\sigma^z_1,\dots,\sigma^z_M,\sigma^x_1,\dots,\sigma^x_M, \HH{I}, \id ,\Hil_M,\hat{\mathcal{A}})$ forms a quantum control system according to Definition~\ref{dfn:con_sys}.
\end{prop}
The control system is presented in a proposition since we still have to show that the operators satisfy the assumptions of a quantum control system (Definition~\ref{dfn:con_sys}). The proof is given in Section~\ref{sec:Rec}. Our main result concerns the controllability of this system.

\begin{thm}[Controllability of the Jaynes-Cummings-Hubbard model]\label{thm:contr_JCH}
	Under the same assumptions as in Proposition~\ref{prop:JCH_con_sys} and for $M\ge2$, the control system
	\begin{equation}
		(\DriftJ,\sigma^z_1,\dots,\sigma^z_M,\sigma^x_1,\dots,\sigma^x_M, \HH{I}, \id ,\Hil_M,\hat{\mathcal{A}})
	\end{equation} is strongly operator controllable.
\end{thm}
The system's free and unsteered evolution is governed by~\eqref{eq:drift_JCH}, i.e.~the JC Hamiltonians of the cavities. Control is exerted on an individual atom $i$ via $\sigma^z_i$ and $\sigma^x_i$. The operator $\HH{I}$ represents collective control of the overall photon hopping rate (we note that assuming a connected hopping interaction graph $I$ is natural since we would consider unconnected graphs as two different separated physical systems). Theorem~\ref{thm:contr_JCH} states that this is sufficient to approximate any unitary operator on the joint atom-photon Hilbert space $\Hil_M$ by a suitably tuned time evolution operator. Obviously, this implies pure-state controllability, i.e.~that any pure states can be approximately interconverted.

Let us briefly outline the proof idea for Theorem~\ref{thm:contr_JCH}. It can be summarized in the following three steps:
\begin{itemize}
	\item Step 1 in Section~\ref{sec:Rec}: First, we check that the system from Proposition~\ref{prop:JCH_con_sys} satisfies the conditions of a quantum control system, i.e., that the operators $\DriftJ$, $\DriftJ+\sigma^z_1$, $\dots$, $\DriftJ+\HH{I}$ and $\DriftJ+\id$ are self-adjoint. We do this by constructing a joint dense domain of essential self-adjointness for all these operators. We also use spectral analysis to show that the system is recurrent, i.e., that the set of reachable unitaries forms a group called the dynamical group.
	\item Step 2 in Section~\ref{sec:2cavities}: We prove strong operator controllability for the JCH model for only two cavities in Theorem~\ref{thm:control_2}. The proof relies on a controllability criterion (Theorem~\ref{thm:sym_strategy} in section 3) which is a generalization of a previous result from \cite{keyl2014controlling}. The criterion exploits a system's symmetry which corresponds for the JCH model to a conservation of the total excitation number. This means that all but one control Hamiltonian conserve the sum of photonic and atomic excitations. For the symmetry obeying control Hamiltonians, we show an inclusion on an increasing sequence of finite dimensional Lie algebras. This is sufficient to uphold strong approximations on the level of the dynamical group. We add the control Hamiltonian $\sigma^x_1$ that breaks the symmetry and achieve strong operator controllability.
	\item Step 3 in Section~\ref{sec:graph}: We generalize the result to an arbitrary number $M\ge2$ of cavities. In order to do so, we interpret the hopping interaction structure of the set $I$ as a graph and use a simple argument on the controllability of overlapping quantum systems.
\end{itemize}

\section{Sufficient conditions for controllability in infinite dimensions}\label{sec:symmetry}
In this section we give a list of sufficient conditions for an infinite dimensional system to be strongly operator controllable.
Let us start with a short comparison to finite dimensional systems where the anti-self-adjoint~drift and control Hamiltonians~$iH_0,\dots, iH_d$ generate (via mutual commutators) a Lie algebra, the so-called system's algebra. The Lie algebra rank condition~\cite{sussmann1972controllability, jurdjevic1972control,brockett1972system,brockett1973lie}, a necessary and sufficient for controllability, states: The system is operator controllable if and only if the system's algebra is isomorphic to the unitary Lie algebra of degree $\dim \Hil$. This condition relates a statement on the level of Lie algebras to an equality between a corresponding Lie group and the set of reachable unitaries. However, the anti-self-adjoint operators on an infinite-dimensional Hilbert space do not even form a Lie algebra. That is why there is no such condition in infinite dimensions and we cannot access controllability questions on the Lie algebra level. Instead, we have to consider the set $\mathcal{R}$. It will be an aim to identify under which conditions it forms a group. But note in general, even if $\mathcal{R}=\mathcal{U}(\Hil)$, it is not a Lie group since $\mathcal{U}(\Hil)$ -- containing all unitary operators on $\Hil$ -- is a topological but not a Lie group with respect to the strong operator topology. 

In this section we make use of a symmetry that all but one control Hamiltonian admit. We establish conditions on finite subsystems, similar to the Lie algebra rank condition. This strategy was introduced by Keyl et al.~in~\cite{keyl2014controlling}. But in contrast to the former, we consider quantum control systems with a non-zero drift Hamiltonian. So we adapt several definitions, introduce the concept of approximate recurrence, and finally generalize the list of sufficient controllability criteria from~\cite{keyl2014controlling} to systems with a non-zero drift. 

\subsection{Dynamical group and recurrence}

\begin{dfn}[Dynamical group]\label{dfn:dyn_group}
	 Let $\Hil$ be a Hilbert space and $A_1,\dots, A_d$ self-adjoint operators. Then the smallest strongly closed subgroup of $\mathcal{U}(\Hil)$ (closed as a subset of $\mathcal{U}(\Hil)$) that contains all $\exp(itA_k)$ for $k\in \{1,\dots,d\}$ and $t\in \mathbb{R}$ is called the dynamical group generated by $A_1,\dots,A_d$. We denote it by $\G(A_1,\dots, A_d)$.
\end{dfn}

\begin{dfn}[Approximate recurrence]\label{dfn:rec}
	A quantum control system $(H_0,H_1,\dots,H_d,\mathcal{H}, \hat{\mathcal{A}})$ satisfies the recurrence condition if and only if for all $H\in \{H_0,H_0+H_1,\dots,H_0+H_d\}$, the following holds: for all $t_-\le0$, all $\epsilon>0$, and all $\ket{\psi_1},\dots,\ket{\psi_f}\in \Hil$ there exists a $t_+\ge0$ such that 
	\begin{equation}
		\big\|\left(e^{it_-H}-e^{it_+H}\right) \ket{\psi_l}\big\|\le \epsilon\quad \text{ for all } \quad l=1,\dots,f.
	\end{equation}
\end{dfn}
The recurrence property is an approximate one. It states that every strong neighborhood of an exponential of the Hamiltonian for some negative time contains another one for some positive time. It allows for studying a system's controllability on the level of a group.

\begin{prop}[Controllability of recurrent systems]\label{prop:suff_cond}
	Let $\mathcal{R}$ be the reachable set of a quantum control system $(H_0,H_1,\dots,H_d,\mathcal{H}, \hat{\mathcal{A}})$ and for all $j=1,\dots,d$, let the operator $H_0+H_j$ be essentially self-adjoint on $\dom(H_0)\cap \dom(H_j)$.
	Then
	\begin{itemize}
		\item $\mathcal{R}=\G( H_0,H_0+H_1,\dots,H_0+H_d )$ and
		\item the system is strongly operator controllable if and only if $ \mathcal{G} (H_0, H_1, \dots,H_d)=\mathcal{U}(\Hil)$.
	\end{itemize} 
\end{prop}
\begin{proof}
	Let $H\in\{H_0, H_0+H_1,\dots, H_0+H_d\}$. The set $\mathcal{R}$ is defined as the strong closure of finite products of $\exp(i t H)$ where $t\ge 0$. Recurrence implies that if $\exp(itH)\in \mathcal{R}$ for some $t\ge 0$, then also $\exp(-itH)\in \mathcal{R}$. Hence $\mathcal{R}$ is the smallest strongly closed subgroup of $\mathcal{U}(\Hil)$ containing all $\exp(itH)$ for $t\in \mathbb{R}$. This  corresponds to the definition of the dynamical group generated by $ H_0,H_0+H_1,\dots,H_0+H_d$ which we denote by $\G_1$ and define $\G_2:= \G (H_0,H_1,\dots,H_d)$. 
	Assume that an operator $K$ is a generator of (w.l.o.g.) $\G_2$. Then one finds two generators of $\G_1$, that are denoted by $K_1$ and $K_2$ and satisfy $\overline{K_1+K_2}=K$. The notation $\overline{K_1+K_2}$ represents the closure of the operator ${K_1+K_2}$.  Since $K_1$ and $K_2$ are self-adjoint operators and $K_1+K_2$ is essentially self-adjoint on the intersection of the operators' domains, Trotter's formula (c.f~Theorem VIII.31 from~\cite{RESI1}) gives
	\begin{equation}
		\exp[it(\overline{K_1+K_2})] ={\text{s--} \lim}_{n\to \infty} \big[\exp(itK_1/n)\exp(itK_2/n)\big]^n .
	\end{equation}
	By definition, the unitaries $\exp(itK_1/n)$ and $\exp(itK_2/n)$ are elements of $\G_1$. Hence for every fixed $n$, the operator $ \big[\exp(itK_1/n)\exp(itK_2/n)\big]^n$ also is. 
	Since a dynamical group is strongly closed by definition, we conclude that $\exp[it(\overline{K_1+K_2})]$ is in $\G_1$. This applies for all generators of $\G_2$ so that we find $\G_2\subseteq \G_1$. 
	Without loss of generality, we can switch the roles of $\G_1$ and $\G_2$ to obtain the desired equality $\G_1=\G_2$. 
	To prove the second statement, remember the definition of strongly operator controllability: a system is strongly operator controllable if and only if $\mathcal{R}=\mathcal{U}(\Hil)$. Combining the first statement of this proposition and $\G( H_0,H_0+H_1,\dots,H_0+H_d )=\mathcal{R}$ yields that $\mathcal{G} (H_0, H_1, \dots,H_d)=\mathcal{U}(\Hil)$ is equivalent to strong operator controllability.
\end{proof}

This proposition can be informally stated as follows: when analyzing the controllability of a recurrent system, one can treat the drift as if it was an additional control Hamiltonian.
A more detailed analysis of recurrence (times) and its relation to controllability of infinite dimensional systems can be found in~\cite{bliss2014quantum} where the authors use the quantum recurrence theorem~\cite{bocchieri1957quantum,wallace2015recurrence}.

Let us add a brief remark why checking recurrence is necessary here, in contrast to the analysis in~\cite{keyl2014controlling}: in the latter there is no drift. Then allowing for negative control function values (-1) is sufficient to ensure that the set of reachable time evolution unitaries forms a group. When adding a drift, this is no longer sufficient as the Hamiltonian would always be of the form $H_0\pm H_i$ but the drift part $H_0$ would never get negative prefactors.

\subsection{Symmetry strategy}

Since this section's central theorem is a generalization of Theorems~2.3 and 2.4 from~\cite{keyl2014controlling}, we start by stating the relevant definitions and results from that article. We motivate and explain them but refer to the original article for proofs. 
The symmetry strategy builds on a $\text{U}(1)$ symmetry that all but one of the control Hamiltonian admit. 
\begin{dfn}[Positive, finitely degenerate charge type operator]\label{def:charge}
	 Let $N$ be a self-adjoint operator on $\Hil$ and $\pi$ a strongly continuous unitary representation of U(1) on $\Hil$ such that $\text{U}(1)\ni z =\exp(i\alpha) \mapsto \pi(z)=\exp(i\alpha N)\in \mathcal{U}(\Hil)$.  Denote by $n$ the eigenvalues of $N$ and by $P_n$ the corresponding eigenprojections. \\
	 Such an operator $N$ is positive, finitely degenerate charge type, if $n\in \mathbb{N}_0$ and all $n$ are of finite multiplicity, i.e.,~if $\dim P_n\Hil< \infty$. 
\end{dfn}
For simplicity, we will call such an operator $N$ ``charge type''. This operator defines a symmetry in the sense that we consider sets of operators commuting with $N$. Define the sets
\begin{align}
	\U (N):&=\{U\in \mathcal{U}(\Hil) \,|\, U \text{ commutes with } N \}\label{eq:UN}\\
	\tilde{\mathfrak{u} } (N):&=\{iH\in \mathcal{L}(\Hil)\,|\, H \text{ self-adjoint and commutes with }N \}\\
	\SU (N):&=\{U\in \U(N) | \det (P_nU P_n)=1 \text{ for all } n\in \mathbb{N}_0 \} \label{eq:SUN} \\
	\su(N):&=\{iH\in \tilde{\mathfrak{u}}(N)\,|\, \operatorname{tr} (P_n\,iH P_n)=0 \text{ for all } n\in \mathbb{N}_0 \}\ .
\end{align} 
It is important to note that $N$ is an operator and not a natural number and hence $\U (N)$ should not be confused with a unitary group of some degree $N$; it is the set of all unitaries on an infinite dimensional Hilbert space $\Hil$ that commute with the operator $N$. 
Also pay attention to another subtlety of the notation: The set $\su(N)$ is not the special subalgebra of $\tilde{\mathfrak{u}} (N)$ as well as $\SU (N)$ is not the special subgroup of $\U (N)$.

Theorem~2.1 of~\cite{keyl2014controlling} shows that $\tilde{\mathfrak{u}} (N)$ forms a Lie algebra with the commutator as its Lie bracket and that it is mapped by the exponential map onto the strongly closed Lie group $\U (N) $. Analogously, $\su (N)$ forms a Lie subalgebra of $\tilde{\mathfrak{u}} (N)$ and mapped by $\exp$ to the Lie group $\SU (N)$. The same theorem also states that if some control Hamiltonians commute with $N$ and satisfy the recurrence property, then they generate a Lie algebra which is mapped by the exponential map to the dynamical group generated by these Hamiltonians.

Let us shortly argue why symmetry prohibits controllability: If the drift and all control Hamiltonians admitted this symmetry, then all reachable unitaries in $\mathcal{R}$ also would. We could at most hope that $\mathcal{R}$ would coincide with $\U(N)$. The system would obviously not be strongly operator controllable since $\U(N)$ is a proper subgroup of $\mathcal{U}(\Hil)$. The relation between dynamical conservation laws and controllability restrictions is discussed in~\cite{turinici2001quantum} for finite dimensional systems. As described above, in infinite dimensions, commutation with a charge type operator also prohibits controllability. In order to reach all unitaries, we would have to add at least one control Hamiltonian that breaks this symmetry. Of course, in order to achieve controllability, this operator has to obey additional properties, motivating the following definition. 

\begin{dfn}[Complementarity]\label{dfn:compl}
Let $N$ be a charge type operator according to Definition~\ref{def:charge}. A self-adjoint operator $H$ with domain $\dom (H)\subseteq \Hil$ is called \textbf{complementary} to $N$ if there exists a decomposition of the Hilbert space $\mathcal{H} = \mathcal{H}_- \oplus \mathcal{H}_0 \oplus \mathcal{H}_+$ such that the following holds:
\begin{itemize}
	\item[(i)] Let $\mathcal{H}_\alpha = E_\alpha \mathcal{H}$ where $\alpha \in \{ +,0,-\}$ and $E_\alpha$ denotes the projections onto the subspaces $\mathcal{H}_\alpha$ and set $P_\pm^{(n)}:=E_\pm P_n$. They satisfy $\left[ E_\alpha,P_n\right]=0$ for all $n\in\mathbb{N}_0$;
	\item[(ii)] We have $\mathcal{H}_0 \subseteq \text{dom}(N)$ and $H \ket{\psi}=0$ for all $\ket{\psi} \in \mathcal{H}_0$;
	\item[(iii)]  We have $\text{dom}(N) \subseteq \text{dom}(H)$ and $H P_+^{(n+1)}\ket{\psi}= P_-^{(n)} H\ket{\psi}$ for all $n>1$ and that $P_-^{(n)} H P_+^{(n+1)}$ is a partial isometry with $P_+^{(n+1)}$ as its source and $P_-^{(n)}$ as its target projection.
	\item[(iv)] Consider the set $\SU(N)$ from Eq.~\eqref{eq:SUN} and the subspace $\mathcal{H}_{\phantom{-}}^{(0)}\oplus\mathcal{H}_+^{(1)}:=\left(P_{\phantom{-}}^{(0)}\oplus P_+^{(1)} \right) \mathcal{H}$ for the projection $P_{\phantom{-}}^{(0)}\oplus P_+^{(1)}$. The group generated by $\exp{(itH)}$ where $t\in \mathbb{R}$ and those $U\in \SU (N)$ which commute with $P_{\phantom{-}}^{(0)}\oplus P_+^{(1)}$ acts transitively on the space of one-dimensional projections in $\mathcal{H}^{(0)}\oplus\mathcal{H}_+^{(1)}$.
\end{itemize}
\end{dfn}

We now combine the results from~\cite{keyl2014controlling} with Proposition~\ref{prop:suff_cond} from the previous subsection and obtain a result that goes beyond prior work. 
We introduce a list of sufficient conditions for an infinite dimensional quantum control system with a non-zero drift to be controllable.

\begin{thm}[Sufficient conditions for controllability]\label{thm:sym_strategy}
	Let $(H_0,H_1,\dots,H_d=\id,\mathcal{H}, \hat{\mathcal{A}})$ be a quantum control system and $N$ a charge type operator. Assume that the following holds:
	 \begin{itemize}
	 	\item[(i)] Self-adjointness: For every $k\in \{1,\dots, d\}$ the operator $H_0+H_k$ is essentially self-adjoint on $\dom(H_0)\cap \dom(H_k)$.
	 	\item[(ii)] Recurrence: The control system satisfies the recurrence property.
	 	\item[(iii)] Symmetry: $H_0, H_1,\dots, H_{d-2}$ commute with $N$.
	 	\item[(iv)] Adapted Lie algebra rank condition: The dynamical group generated by $H_0, H_1,\dots, H_{d-2}, \id$ contains the group $\SU(N)$ from Eq.~\eqref{eq:SUN}. 
	 	\item[(v)] Breaking the Symmetry: $H_{d-1}$ is complementary to $N$. 
	 \end{itemize}
	 Then the control system is pure-state controllable. If additionally 
	 \begin{itemize}
	 	\item[(vi)] $\dim P_n\Hil >2$ for at least one $n\in \mathbb{N}_0$
	 \end{itemize}
	 holds, then it is also strongly operator controllable.
\end{thm}
The only difference to Theorems~2.3 and 2.4 from~\cite{keyl2014controlling} is that here we have a non-zero drift term and the first two additional assumptions. 
\begin{proof}
	The set of reachable unitaries $\mathcal{R}$ is due to recurrence and Proposition~\ref{prop:suff_cond} equal to the dynamical group generated by $\id, H_0, H_1,\dots, H_{d}$. From this point on, we can treat the control system as if $H_0$ was one of the control Hamiltonians and follow the proof of Theorems 2.3 and 2.4 from \cite{keyl2014controlling} until we get $\mathcal{G}(H_0,\id,H_1,\dots,H_d)=\mathcal{U}(\Hil)$. With Proposition~\ref{prop:suff_cond} this gives that the system is strongly operator controllable.
\end{proof}

Note that the `adapted Lie algebra rank condition' is stated here on the level of Lie groups. But we will give an equivalent condition on the algebra level motivating the name `Lie algebra rank condition': Proposition~\ref{prop:suff_trunc_compl} states that it is sufficient to show an inclusion of finite dimensional Lie subalgebras of $\su (N)$. Its proof mainly relies on Theorem~2.1 from~\cite{keyl2014controlling} (stating that the group $\SU(N)$ as well as the dynamical group of operators commuting with $N$ are generated by the corresponding Lie algebras).
But let us first introduce some additional notation.
Given a charge type operator $N$ with eigenprojections $P_n$, we can decompose $\Hil$ into 
\begin{equation}
	\Hil=\bigoplus_{n=0}^\infty  P_n\Hil \ . 
\end{equation}
Due to $\dim P_n\Hil<\infty$ we can define $d_n:=\dim P_n\Hil$ and the  subspaces
\begin{equation}\label{eq:notation_decomp_hil}
	\Hil^{(n)}:=P_n\Hil\ , \qquad \Hil^{[K]}:= \bigoplus_{n=0}^K \Hil^{(n)}\ 
\end{equation} are finite dimensional. For an operator $A$ commuting with $N$, let
\begin{equation}\label{eq:notation_decomp_op}
	A^{(n)}:=P_n A P_n\ ,\qquad  A^{[K]}:=\sum_{n=0}^{K} P_n A P_n
\end{equation} denote operators on these finite dimensional subspaces.

\begin{prop}[`Adapted Lie algebra rank condition' on finite dimensional Lie subalgebras]\label{prop:suff_trunc_compl}
	Let $(H_0,H_1,\dots,H_d=\id,\mathcal{H}, \hat{\mathcal{A}})$ be a quantum control system and $N$ a charge type operator. Assume that conditions $(i)$ -- $(iii)$ from Theorem~\ref{thm:sym_strategy} hold and use the notation from Eq.~\eqref{eq:notation_decomp_op}. For some $K\in \mathbb{N}$, let $\lC$ and $\lCK$ denote the complexifications of the Lie algebras 
	\begin{equation}
		\mathfrak{l}:=\Lie{iH_0,\dots,iH_{d-2}, i\id} \qquad \text{and}\qquad \mathfrak{l}^{[K]}:=\Lie{i H_0^{[K]},\dots,i H_{d-2}^{[K]}, i\id^{[K]} }\ ,
	\end{equation} respectively. Let $\mathfrak{sl}(d_n;\mathbb{C})$ be the special linear Lie algebra of degree $d_n$, as a vector space over $\mathbb{C}$.\\
	If for all $K\in \mathbb{N}_0$ we have $ \bigoplus_{n=0}^{K}\mathfrak{sl}(d_n; \mathbb{C}) \subseteq \lCK$ then $\SU(N) \subseteq \G$. 
\end{prop} 
\begin{proof}
	Applying Propositions~4.4 and 4.5 from~\cite{keyl2014controlling} with the identifications $\su(N)=:\mathfrak{l}_1$ and $\mathfrak{l}=:\mathfrak{l}_2$ yields: if for all $K\in \mathbb{N}_0$ 
	\begin{equation}
		\su^{[K]}(N):=\big\{A^{[K]} \big| \,A\in\su(N) \big\}=\bigoplus_{n=0}^K \,\mathfrak{su} \big(\Hil^{(n)} \big)\subseteq \mathfrak{l}^{[K]}
	\end{equation}then the Lie group $\SU(N)$ is contained in the dynamical group $\G(H_0, H_1,\dots, H_{d-2}, \id)$. 
	Furthermore, it is sufficient to prove this inclusion for the complexified algebras $ \su^{[K]}_\mathbb{C}(N)=\su^{[K]}(N)\oplus i \,\su^{[K]}(N) $ and $\mathfrak{l}^{[K]}_\mathbb{C}$ since we get the original statement back by restricting to anti-self-adjoint operators on both sides. The complexification of $\su^{(n)}(\Hil^{(n)})$ is isomorphic to the Lie algebra $\mathfrak{sl}(d_n;\mathbb{C}) $ of traceless $d_n\times d_n$-matrices with entries in $\mathbb{C}$.
	Summing over $n$ from $0$ to $K$ gives  $\su^{[K]}_\mathbb{C}(N)=\bigoplus_{n=0}^{K}\mathfrak{sl}(d_n; \mathbb{C})$ and concludes the proof.
\end{proof}

\section{Spectral analysis for the JCH model}\label{sec:spectral_self}

In this section, we study the Jaynes-Cummings-Hubbard model with $M\ge2$ cavities. We provide a proof for Proposition~\ref{prop:JCH_con_sys} defining the JCH quantum control system. We additionally find a charge-type operator that commutes with all of the unbounded control Hamiltonians. We use it to construct a complete orthonormal basis of eigenvectors of the Hamiltonians and show recurrence of the control system. 
All proofs rely on spectral analysis. The key step is to construct a joint domain of analytic vectors for all relevant operators. 

\subsection{Self-adjointness and control system}
An analytic vector $\ket{\psi}$ of an operator $A$ is defined as an element of $\cap_{m=1}^\infty \dom(A^m)$ that satisfies 
\begin{equation}
	\sum_{m=0}^\infty \frac{t^m \|A^m \ket{\psi} \|}{m!}<\infty\; 
\end{equation} for some $t>0$. On such vectors, the power series of $\exp(itA) \ket{\psi}$ makes sense for sufficiently small $t$. They can be used to check whether an operator is essentially self-adjoint, i.e., whether its closure is self-adjoint.
For further details we refer to textbooks like chapters VIII from~\cite{RESI1} and X from~\cite{RESI2}.
We will show that the set $\fset$ from Eq.~\eqref{eq:D} forms a so called total set of analytic vectors for all relevant operators, i.e., it is a dense set of analytic vectors and invariant under the operators' action. Hence these operators are essentially self-adjoint on $\fset$. 

\begin{lem}[Domain of essential self-adjointness] \label{lem:selfad}
	The operators $\id$, $\sigma^z_i$, $\sigma^x_i$, $a^*_ia^{\phantom{*}}_i$, $a^*_i \sigma^-_i + a^{\phantom{*}}_i \sigma^+_i$, and $a^*_i a^{\phantom{*}}_j + a^{\phantom{*}}_i a^*_j $ for $1\le i,j\le M$ and $i\neq j$ admit $\fset$ as a total set of analytic vectors. 
	Arbitrary real linear combinations of them are essentially self-adjoint on $\fset$.
\end{lem}
The proof is a simple application of Nelson's analytic vector theorem (e.g.~Theorem X.39 from~\cite{RESI2}) and uses properties of the Hermite functions. Its components are not new but we did not find a reference where they were applied to the operators at hand and will hence briefly do this here. 
\begin{proof}
	We first consider the above unbounded operators and show that they are symmetric on supersets of $\fset$ which they leave invariant.
	For $i= 1,\ldots,M$ the operators 
	\begin{equation}\label{eq:symmetricoperators}
		a^*_ia^{\phantom{*}}_i \ , \quad a^{\phantom{*}}_i+a^*_i\ ,\quad ia^{\phantom{*}}_i-ia^*_i
	\end{equation} are symmetric on domains which contain $\fset$. For $n\in\mathbb{N}_0$ let $\ket{n}$ be the Hermite functions in number basis; the set 
	\begin{align}\label{eq:selfad_1}
		\text{span} \{ \ket{\psi_1} \otimes \ket{n_i}\otimes \ket{\psi_2} \in \Hil_M \ \big|\ & \ket{\psi_1} \in\Hil_{i-1}, \ket{\psi_2} \in\mathbb{C}^2\otimes \Hil_{M-i}, n_i\in \mathbb{N}_0\}
	\end{align}
	is invariant under the action of the operators~\eqref{eq:symmetricoperators} and hence its subset $\fset$ also is. Assuming w.l.o.g.~$1\le i<j\le M$, the symmetric operator $a^*_i a^{\phantom{*}}_j  +a^{\phantom{*}}_i a^*_j$ leaves the set 
	\begin{align*}
		\text{span} \big\{\ket{\psi_1} \otimes\ket{n_i}\otimes \ket{\psi_2} \otimes \ket{n_j}\otimes \ket{\psi_3} \in \Hil_M \ \big|\ & \ket{\psi_1} \in\Hil_{i-1}, \ket{\psi_2} \in \mathbb{C}^2 \otimes \Hil_{j-i-1},\\
		&\;\ket{\psi_3}\in\mathbb{C}^2\otimes \Hil_{M-j}, n_i, n_j\in \mathbb{N}_0 \big\}	
	\end{align*}
	and its subset $\fset$ invariant. Consider $a\otimes \sigma_+ + a^*\otimes \sigma_-=\tfrac{1}{2} \left(a+a^*\right)\otimes \sigma^x+ \tfrac{1}{2} \left(ia-ia^*\right)\otimes \sigma^y $ for the operator $a^*_i \sigma^-_i + a^{\phantom{*}}_i \sigma^+_i$. Its summands are tensor products of the symmetric operators $a+a^*$ or $ia-ia^*$ with the self-adjoint bounded Paulis $\sigma^x$ or $\sigma^y$, respectively. They are again symmetric and the set \eqref{eq:selfad_1} as well as $\fset$ is invariant under their action. In summary, the set $\fset$ is a joint dense invariant domain of $\id$, $\sigma^z_i$, $\sigma^x_i$, $a^*_ia^{\phantom{*}}_i$, $a^*_i \sigma^-_i + a^{\phantom{*}}_i \sigma^+_i$, and $a^*_i a^{\phantom{*}}_j + a^{\phantom{*}}_i a^*_j $ for $1\le i,j\le M$.
	
	We now show that it consists of analytic vectors for the above operators. 
	It is well known that the Hermite functions (as well as their span) are analytic vectors for linear combinations of $a$ and $a^*$ such as $a+a^*$ and $ia-ia^*$ (e.g. Example~2 to Theorem X.39 from \cite{RESI2}). Twofold tensor products of them form a dense set of analytic vectors also for quadratic terms in $a$ or $a^*$ such as $a^{\phantom{*}}_ia^*_j$ or $a^*_ia^{\phantom{*}}_j$ (e.g. Proposition 4.49 from \cite{folland1986harmonic}). 
	Hence $\fset$ forms a dense set of analytic vectors for the operators $a^*_ia^{\phantom{*}}_i$, $a^*_i+a^{\phantom{*}}_i $, $ia^*_j- i a^{\phantom{*}}_j$, and $a^*_i a^{\phantom{*}}_j +a^{\phantom{*}}_i a^*_j$. 
	Combining both properties of $\fset$, invariance and analyticity, 
	we can apply Nelson's analytic vector theorem (Theorem X.39 from \cite{RESI2}) to conclude that these operators as well as linear combinations of them are essentially self-adjoint on $\fset$.
	Note that $a^*_i \sigma^-_i + a^{\phantom{*}}_i \sigma^+_i$ inherits this property from $a^*_i+a^{\phantom{*}}_i $ and $ia^*_j- i a^{\phantom{*}}_j$ since Theorem VIII.33 from \cite{RESI1} allows for analyzing essential self-adjointness on factors of a tensor product separately. 
	
	The bounded and symmetric operators $\id$, $\sigma^z_i $, and $\sigma^x_i$ are self-adjoint on $\Hil$. 
	Note that due to Wüst's theorem (Theorem X.14 from \cite{RESI2}), the sum of a bounded self-adjoint and an essentially self-adjoint operator on some domain is again essentially self-adjoint on the same domain. Hence arbitrary linear combinations of the operators $\id$, $\sigma^z_i$, $\sigma^x_i$, $a^*_ia^{\phantom{*}}_i$, $a^*_i \sigma^-_i + a^{\phantom{*}}_i \sigma^+_i$, and $a^*_i a^{\phantom{*}}_j + a^{\phantom{*}}_i a^*_j $ are essentially self-adjoint on $\fset$.
\end{proof}

This Lemma implies that the drift and all control Hamiltonians for the JCH model as well as their linear combinations are essentially self-adjoint on $\fset$ and hence have unique self-adjoint extensions.

\begin{proof}[Proof of Proposition~\ref{prop:JCH_con_sys}]
	The operators $\DriftJ$, $\DriftJ+\sigma^x_i$, $\DriftJ+ \sigma^z_i$ (for $i\in \{1,\dots, M \}$) and $\DriftJ+ \HH{I} $ are linear combinations of the operators from Lemma~\ref{lem:selfad}. Hence they are essentially self-adjoint on the set $\fset$ and we can, in slight abuse of notation, denote their self-adjoint extensions by the same symbol. We conclude that the drift and control Hamiltonians from Proposition~\ref{prop:JCH_con_sys} form a quantum control system according to Definition~\ref{dfn:con_sys}.
\end{proof}

\subsection{Symmetry and charge type operator}
To show recurrence we exploit the system's symmetry which we also need to apply Theorem \ref{thm:sym_strategy}. 
\begin{lem}[Charge-type operator]\label{lem:N_chargetype}
	 The operator 
	 \begin{equation}
		 N:=\sum_{i=1}^M \left(a^*_i a^{\phantom{*}}_i+ \tfrac{1}{2}\left(\sigma^z_i + \id^{\phantom{*}}_i \right) \right)
	 \end{equation} is a charge type operator according to Definition~\ref{def:charge}.
\end{lem}
\begin{proof}
	$N$ is self-adjoint due to Lemma~\ref{lem:selfad}. It acts via
	\begin{align}
		N  \ket{m_1}\otimes &\ket{b_1}\otimes \cdots \otimes \ket{m_M}\otimes \ket{b_M} \nonumber\\
		&=(m_1+b_1+\cdots+m_M+b_M) \ket{m_1}\otimes \ket{b_1}\otimes \cdots \otimes \ket{m_M}\otimes \ket{b_M} 
	\end{align} on basis vectors of $\Hil_M$. Its eigenvalues are of the form $\{\sum_{i=1}^M(m_i+b_i)| m_i\in\mathbb{N}_0,b_i\in\{0,1\}\} =\mathbb{N}_0$. Let us denote them by $n\in \mathbb{N}_0$ and the projections onto their corresponding eigenspace by $P_n$. Their multiplicity is bounded by $\dim (P_n \Hil_M)\le 2^M (n+1)^{M-1}<\infty$. Hence $N$ is charge-type.
\end{proof}

The eigenvalues of $N$ correspond to the sum of the photonic and atomic excitations and are called the total excitation numbers. The operator $N$ can be written as $N=\sum_{i=1}^{M}N_i$ where
\begin{equation}
	N_i:= a^*_i a^{\phantom{*}}_i+ \tfrac{1}{2}\left(\sigma^z_i + \id^{\phantom{*}}_i \right)
\end{equation} which in turn measure the (sum of photonic and atomic) excitation number in cavity $j$. We introduce the new notation 
\begin{equation}
	\ket{n;n_1;\dots; n_{M-1};b_1,\dots ,b_M}: =\ket{n_1-b_1}\otimes \ket{b_1}\otimes\ket{n_2-b_2}\otimes \cdots\otimes \ket{b_M}.
\end{equation} for the natural basis $\ket{m_1} \otimes \ket{b_1}\otimes\cdots\otimes \ket{m_M} \otimes\ket{b_M}$ of $\Hil_M$.  
This allows for rewriting the set
\begin{align}
	\fset& :=\text{span} \{ \ket{m_1}\otimes \ket{b_1}\otimes \cdots \otimes \ket{m_M}\otimes \ket{b_M} |b_i\in\{0,1\},\, m_i\in \mathbb{N}_0,\, i=1,\dots, M\}\nonumber\\
	&=\text{span}\{\ket{n;n_1;\dots; n_{M-1};b_1,\dots ,b_M}| n_i-b_i\ge 0,b_i\in \{0,1\}, 0\le n\le n_1+\dots n_{M-1} +b_M  \}\nonumber\\
	&= \text{span}\{ \ket{\psi} \in \Hil_M| \exists n\in \mathbb{N}_0, \, N\ket{\psi}=n \ket{\psi} \} \label{eq:D_2}
\end{align} and justifies its name as the finite excitation space.

\begin{prop}[Symmetry]\label{lem:commutation}
	The operators $\DriftJ$, $\HH{I}=\sum_{(j,k)\in I} (a^*_j a^{\phantom{*}}_k + a^{\phantom{*}}_j a^*_k)$, and $\sigma^z_i$ for $i\in \{1,\dots, M \}$ commute with the total particle operator $N$; the operators $N$, $\DriftJ$ and $\sigma^z_i$ also commute with $N_j$ for all $j\in \{1,\dots,M\}$.
\end{prop}
\begin{proof}
	Two unbounded operators $A$ and $B$ commute if and only if all their spectral projections commute. Note that in contrast to popular belief it is not sufficient that they satisfy $AB \ket{\psi}=BA\ket{\psi} $ on a joint dense domain of essential self-adjointness~(\cite{nelson1959analytic}; section VIII.5 of~\cite{RESI1}). However, if $AB\ket{\phi}=BA \ket{\phi}$ on a total set of common analytic vectors then the two operators commute (c.f~Theorem~7.18 of~\cite{schmudgen2012unbounded}). 
	The set $\fset$ from \eqref{eq:D_2} forms a total set of analytic vectors for all relevant operators due to Lemma~\ref{lem:selfad}. Defining  $n_M:=n-(n_1+n_2+\cdots+ n_{M-1})$ we calculate
	\begin{align}
		& N \ket{n;n_1,\dots,b_M} = n \ket{n;n_1,\dots,b_M}\\
		& N_i \ket{n;n_1,\dots,b_M}= n_i \ket{n;n_1,\dots,b_M}\\
		& \sigma^z_i \ket{n;n_1,\dots,b_M}= \left(b_i-\tfrac{1}{2} \right) \ket{n;n_1,\dots,b_M}\\
		& a^*_ia^{\phantom{*}}_i \ket{n;n_1,\dots,b_M}= \left(n_i-b_i\right) \ket{n;n_1,\dots,b_M}\; .
	\end{align}
	One has $(a^*_i \sigma^-_i + a^{\phantom{*}}_i \sigma^+_i )\ket{n;n_1,\dots,b_i,\dots, b_M}=\sqrt{n_i} \ket{n;n_1,\dots,1-b_i, \dots, b_M}$ if $n_i\ge 2$ or if $n_i=1$ and $b_1=1$; one has $(a^*_i \sigma^-_i + a^{\phantom{*}}_i \sigma^+_i)\ket{n;n_1,\dots,b_i,\dots, b_M}=0$ otherwise. By linearity, the operators $\sigma^z_i$, $a^*_ia^{\phantom{*}}_i $, $a^*_i \sigma^-_i + a^{\phantom{*}}_i \sigma^+_i $ and $\DriftJ$ commute both with $N$ and $N_i$ on $\fset$. Hence their self-adjoint extensions commute.
	We find 
	\begin{align*}
		(a^*_i a^{\phantom{*}}_j + a^{\phantom{*}}_i a^*_j ) \ket{n;\dots,b_M} = \sqrt{n_i-b_i+1} \sqrt{n_j-b_j} \ket{n; \dots, n_i+1, \dots, n_j-1,\dots b_M} + i \leftrightarrow j
	\end{align*} if $n_i-b_i\ge1$ and $n_j-b_j\ge 1$. If $n_i-b_i\ge1$ but $n_j-b_j= 0$ we get:
	\begin{align}\label{eq:commutation_proof_1}
		(a^*_i a^{\phantom{*}}_j + a^{\phantom{*}}_i a^*_j ) \ket{n;\dots,b_M} = \sqrt{n_i-b_i+1}\sqrt{n_j-b_j} \ket{n; \dots, n_i+1, \dots, n_j-1,\dots b_M}\; .
	\end{align} If instead $n_j-b_j\ge1$ but $n_i-b_i= 0$ we get the same as in Eq.~\eqref{eq:commutation_proof_1} just replacing $i\leftrightarrow j$. We obtain $(a^*_i a^{\phantom{*}}_j + a^{\phantom{*}}_i a^*_j ) \ket{n;\dots,b_M}=0$ otherwise.
	Hence the operators $\HH{I}=\sum_{(i,j)\in I} (a^*_i a^{\phantom{*}}_j + a^{\phantom{*}}_i a^*_j)$ and $N$ commute on $\fset$ and we conclude that their self-adjoint extensions commute as well.
\end{proof}

\subsection{Recurrence}

In order to shorten notation, introduce the free photon operator $\HP$, the free atom $\HA$, and the photon-atom interaction $\HI$ given by
\begin{align}\label{eq:notation_HP_HI}
	\HP&:=\sum_{k=1}^M \oP{k} \, a^*_k a^{\phantom{*}}_k\qquad \HA:=\sum_{k=1}^M \oA{k} \, \sigma^z_k\qquad \HI:=\sum_{k=1}^M \oI{k} \, (a^*_k \sigma^-_k + a^{\phantom{*}}_k \sigma^+_k )\; .
\end{align}

We now prove that the quantum control system from Proposition~\ref{prop:JCH_con_sys} satisfies the recurrence condition. To do so, we construct a complete orthonormal set of eigenvectors of the relevant Hamiltonians $\DriftJ$, $\DriftJ+\sigma^z_i$, $\DriftJ+\sigma^x_i$, (for $i\in \{1,\dots, M \}$) and $\DriftJ+\HH{I}$. This is quite simple for the ones commuting with $N$. For $\DriftJ+\sigma^x_i$, which do not commute with $N$, we use properties of their so-called resolvent. For a closed operator $T$ and for those $\lambda\in\mathbb{R}$ such that $\lambda\id-T$ is a bijection of $\dom(T)$ onto the whole Hilbert space $\Hil$ with bounded inverse, one can define the resolvent of $T$ as $R_\lambda(T)=(\lambda\id-T)^{-1}$. $T$ having compact resolvent is equivalent to the statement that in $\dom(T)$, there is an orthonormal basis of eigenvectors of $T$ with increasing eigenvalues $\lambda_1\le\lambda_2\le\dots$ where $\lambda_n\to \infty$ (c.f~Theorem~XIII.64 from~\cite{RESI4}). 
Hence we first prove that $\HP$ has compact resolvent and then that $\DriftJ+\sigma^x_i=\HP+\HA+\HI+\sigma^x_i$ shares this property. 

\begin{lem}[Complete orthonormal set of eigenvectors for commuting operators]\label{lem:spec_com}
	The operators $\DriftJ$, $\DriftJ+\HH{I}$ and $\DriftJ+\sigma^z_i$ for $i\in \{1,\dots, M \}$ admit a complete orthonormal set of eigenvectors in $\Hil$.
\end{lem}
\begin{proof}
	Let $H\in \{\DriftJ,\DriftJ+\sigma^z_i,\DriftJ+\HH{I}\}$ where $i\in \{1,\dots, M \}$. We consider the direct sum decomposition of the Hilbert space and the notation with upper indices $[K]$ from Eqs.~\eqref{eq:notation_decomp_op} and~\eqref{eq:notation_decomp_hil}.
	Remembering that $H$ is essentially self-adjoint on $\fset$ and commutes with $N$ gives that for any $\ket{\psi} \in \fset$ there exists a $K\in \mathbb{N}_0$ such that $H^{[K]}\ket{\psi}=H \ket{\psi}$. The operator $H^{[K]}$ is self-adjoint on $\Hil_M^{[K]}$ and its eigenvectors hence form an orthonormal basis of this finite dimensional space. For $K'\ge K$, the above eigenvectors are also eigenvectors of $H^{[K']}$ and we can extend this set to an eigenbasis of $H^{[K']}$. Taking $K$ arbitrarily large yields a complete, orthonormal set of eigenvectors of $H$.
\end{proof}

\begin{lem} 
	The operator $\HP$ has compact resolvent.
\end{lem}
\begin{proof}
	It is easy to see that $\ket{n; n_1,\dots, b_M}$ are eigenvectors of $\HP$ that correspond to eigenvalues $ \sum_{i=1}^{M} \oP{i} \left( n_i-b_i\right)\ge 0$ where $b_i\in\{0,1\}$ and $n_i\ge b_i\in \mathbb{N}_0$. The eigenvalues can be labeled by $\lambda_k$ where $k=1,2\dots$ and ordered such that $\lambda_k\le \lambda_{k+1}$ and $\lim_{k\to \infty} \lambda_k =\infty$. The corresponding eigenvectors form an orthonormal basis of $\Hil_M$ and therefore $H_P$ has compact resolvent (c.f~Theorem XIII.64 from~\cite{RESI4}).
\end{proof}

\begin{lem}[Compact resolvent of non-commuting operators]\label{lem:spec_non}
	The operators $\DriftJ+\sigma^x_i$ where $i\in \{1, \dots, M \}$ have compact resolvent.
\end{lem}
The proof is divided into three steps: First, we examine $\HP+\HI$. 
We show that on $\fset$, the operator $\HI$ is relatively $\HP$-bounded with relative bound strictly smaller than 1; In step 2 we add a bounded operator $\HA+\sigma^x_i$ with non-trivial action only on the atom Hilbert spaces $\mathbb{C}^2$ and show that a similar relation holds for the self-adjoint extensions of $ \HI+ \HA +\sigma_i^x$ and $\HP$. Lastly we conclude from the relative formbounds that the operators $\DriftJ+\sigma^x_i=\HP+\HI+\HA+\sigma^x_i$ have compact resolvent for all $i\in \{1, \dots, M \}$.

\begin{proof}[Proof of Lemma~\ref{lem:spec_non}]
	Let $\ket{\phi} \in \fset$. Then there exists $K\in \mathbb{N}_0$ such that 
	\begin{equation}
		\ket{\phi}= \sum\limits_{n, \dots, b_M} \phi_{n,\dots,b_M}\ket{n;\dots, b_M}
	\end{equation} where we introduced the sum
	\begin{equation}
		\sum_{n, \dots, b_M}:=\sum_{n=0}^{K} \sum_{n_1=0}^{n} \sum_{n_2=0}^{n-n_1}\dots \sum_{n_{M-1}=0}^{n-n_1-\dots-n_{M-2}} \sum_{b_1=0}^{\min \{1,n_1\}}\dots \sum_{b_M=0}^{\min\{1,n-n_1-\dots-n_{M-1}\} } .
	\end{equation}
	\textbf{Step 1:} Assume that $b,c\in \mathbb{R}$ and consider the term
	\begin{align}
		\|( b \HP+c)\ket{\phi} \|^2 & =\Big\|\sum_{n,\dots, b_M} \left(\sum_{i=1}^M a \, \oP{i} \,n_i +c \right) \phi_{n,\dots, b_M} \ket{n;n_1,\dots, b_M}  \Big\|^2 \nonumber\\
		& =\sum_{n,\dots, b_M} \Big| b \left(\sum_{i=1}^M \oP{i}\, n_i+c \right) \phi_{n,\dots, b_M}\Big|^2\; , \label{eq:spec_non_1}
	\end{align}
	that we want to find a lower bound on. Inserting the notation $\mu:=\min_i \oP{i} $ into Eq.~\eqref{eq:spec_non_1} gives
	\begin{align}
		\| \left( b \HP+c\right)\ket{\phi} \|^2& \ge \mu^2 \left( b K+\tfrac{c}{\mu}\right)^2 \sum_{n,\dots, b_M}| \phi_{n,\dots, b_M}|^2 \;.
	\end{align}
	We also want to find an upper bound on $\| \HI \ket{\phi} \|^2$. We compute
	\begin{align}
		\| \HI \ket{\phi} \|^2 & =\Big\| \sum_{i=1}^M \sum_{n, \dots, b_M} \phi_{n, \dots, b_M} \, \oI{i} \,\sqrt{n_i} \,\ket{n;n_1,\dots, 1-b_i, \dots, b_M} \Big\|^2 \nonumber\\
		& = \Big\| \sum_{i=1}^M \oI{i} \sum_{n, \dots, b_M} \phi_{n,\dots, 1-b_i, \dots, b_M} \sqrt{n_i} \, \ket{ n;n_1,\dots, b_i, \dots, b_M}\Big\|^2 \nonumber\\
		&=\sum_{n, \dots, b_M} \Big| \sum_{i=1}^M  \oI{i} \, \phi_{n,\dots, 1-b_i, \dots, b_M} \sqrt{n_i} \Big|^2\; .\label{eq:spec_non_2}
	\end{align}
	In the second line, we replaced the finite sum $\sum_{n, \dots, b_M}$ by an adapted sum over the same symbols, substituting $b_i$ with $1-b_i$ for every $i=1,\dots,M$. But since $b_i\in \{0,1\}$, both sums are identical. Defining $\lambda:=\max_i \oI{i}$ in Eq.~\eqref{eq:spec_non_2}, we can use the above sum replacement again to conclude that
	\begin{align}
		\| \HI \ket{\phi} \|^2& \le  \lambda^2\sum_{n, \dots, b_M} \Big| \sum_{i=1}^M \phi_{n,\dots, 1-b_i, \dots, b_M} \sqrt{n_i} \Big|^2 
		\le \lambda^2 M K \sum_{n, \dots, b_M}|\phi_{n, \dots, b_M}|^2 \;.\label{eq:spec_non_3}
	\end{align}
	It is easy to check that we can always choose real constants $c>0$ and $0< b < 1$ such that the maximum of $f(K)= \lambda^2 M K-\mu^2(bK + c/\mu )^2$ is smaller than 0, and hence $\lambda^2 M K\le \mu^2(aK + c/\mu )^2$ for all $K\in \mathbb{N}$. Combining this with Eqs.~\eqref{eq:spec_non_2} and \eqref{eq:spec_non_3} gives
	\begin{align*}
		\| \HI \ket{\phi} \|^2 & \le \lambda^2 M K \sum_{n, \dots, b_M}|\phi_{n, \dots, b_M}|^2 \le \mu^2 \left(bK + \tfrac{c}{\mu} \right)^2  \sum_{n, \dots, b_M}|\phi_{n, \dots, b_M}|^2 \le \| \left( b \HP+c\right)\ket{\phi} \|^2 .
	\end{align*}
	Hence there exist constants $c>0$ and $0< b< 1$ such that $\| \HI \ket{\phi} \|\le b\| \HP \ket{\phi}\|+ c\|\ket{\phi} \|$ for all $\ket{\phi} \in \fset$.\\
	\textbf{Step 2}: In Lemma~\ref{lem:selfad}, we showed that $\fset$ is a total set of analytic vectors for the operators $\HP$, $\HI$, $\HA$ and $\sigma^x_i$. Using basic properties of the closure, we can extend the result from step~1 to their self-adjoint extensions.
	Due to the boundedness of $\HA$ and $\sigma^x_i$, there exist constants $c>0$ and $0<b<1$ such that
	$\| (\HI+\HA+\sigma^x_i)\ket{\psi}\|\le b \|\HP \ket{\psi} \|+ (c+\|\HA \|+\|\sigma^x_i\|)\| \ket{\psi}\|$ for all $\ket{\psi}\in \dom(\HP)$. Since $\HP$ is positive, we apply Theorem X.18 from~\cite{RESI2} to find constants $\tilde{c}>0$ and $0<b<1$ such that 
	\begin{equation}
		|\langle \psi, (\HI+\HA+\sigma^x_i ) \psi\rangle |\le b \langle \psi, \HP \psi\rangle+ \tilde{c} \langle\psi,\psi \rangle 
	\end{equation} for all $\ket{\psi}$ in the form domain of $\HP$.\\
	\textbf{Step 3} is an application of Theorem XIII.68 from \cite{RESI4}. We showed that $\HP$ is a self-adjoint semi-bounded operator with compact resolvent and $\HI+\HA+\sigma^x_i$ is a symmetric form-bounded perturbation with relative form-bound strictly smaller than 1. Then we can do the same as in the proof of Theorem XIII.68 from \cite{RESI4} to conclude that the sum $H_P+\HI+\HA+\sigma^x_i$ has compact resolvent.
\end{proof}  

\begin{thm}[Recurrence of the JCH control system]\label{thm:rec_JCH}
	The quantum control system from Proposition~\ref{prop:JCH_con_sys} satisfies the recurrence property.
\end{thm}
\begin{proof}
	Let $H\in \{\DriftJ,\DriftJ+\sigma^z_i, \DriftJ+ \sigma^x_i, \DriftJ+\HH{I} |i=1, \dots,M \}$. Lemmas~\ref{lem:spec_com} and~\ref{lem:spec_non} yield that for such $H$ there exists an orthonormal basis of eigenvectors $\ket{\psi_k}$ with eigenvalues $\lambda_k$ and $k=1,2\dots$. Let $t_-\le 0$, and $\epsilon > 0$ and for some $ f\in \mathbb{N}$ let $\ket{\phi_1},\dots,\ket{\phi_f}\in \Hil_M$. Then we can find a $K\in \mathbb{N}$ such that $\|(\id-Q_K) \ket{\phi_j}\|\le \epsilon /3 $ for all $j=1,\dots,f$ where $Q_K$ denotes the projection onto $\text{span} \{\ket{\psi_1},\dots, \ket{\psi_K}\}$. The space $Q_K\Hil_M$ is finite dimensional and invariant under $H$.  For $s\in \mathbb{R}$, the expression $U_{K,s}:=\exp(isH)|_{Q_K\Hil_M}$ defines a strongly-continuous one parameter group on $Q_K\Hil_M$ which maps the real parameter $s$ to a one-dimensional Lie subgroup of $\text{U}(Q_K\Hil_M)$. Hence in $Q_K\Hil_M$ the recurrence condition that we need is satisfied: we can find a $t_+\in \mathbb{R}$, with $t_+\ge 0$ such that $\|(U_{t_+}-U_{t_-})Q_K \ket{\phi_j} \|\le \epsilon /3$ for $j=1,\dots,f$. Thus calculate
	\begin{align*}
		\left\| \big(e^{it_+H}- e^{it_-H} \big)\ket{\phi_j}\right\| &\le \left\| \big(e^{it_+H}- e^{it_-H} \big)Q_K \ket{\phi_j} \right\| +\left\| \big(e^{it_+H}- e^{it_-H} \big)(\id-Q_K)\ket{\phi_j} \right\| \\
		&\le \left\| \big(U_{K,t_+}-U_{K,t_-} \big)Q_K \ket{\phi_j}\right\|+\left\| e^{it_+H}(\id -Q_K) \ket{\phi_j}\right\| +\left\| e^{it_-H} (\id-Q_K) \ket{\phi_j}\right\| \nonumber \\ 
		&\le \frac{\epsilon}{3}+\left\| e^{it_+H}\right\| \left\|(\id-Q_K) \ket{\phi_j}\right\| +\left\| e^{it_-H}\right\|  \left\| (\id-Q_K) \ket{\phi_j}\right\| \\
		&\le \frac{\epsilon}{3}+\frac{\epsilon}{3}+\frac{\epsilon}{3}=\epsilon \; .
	\end{align*}
	This shows that for every $t_-\le 0$ and every strong $\epsilon$--neighborhood $V_\epsilon$ of $\exp(it_-H)$ we can find a $t_+\ge 0$ such that $\exp(it_+H)$ is in $V_\epsilon$.
\end{proof}

\section{Controllability of JCH model with two cavities}\label{sec:2cavities}

Here we show that the control system from Proposition~\ref{prop:JCH_con_sys} for $M=2$ cavities is strongly operator controllable. This will later serve as an induction basis for the general statement in Theorem~\ref{thm:contr_JCH}. In the two cavity case, we use the symmetry strategy from Theorem~\ref{thm:sym_strategy}. 

\subsection{Notation and proof strategy}

Let indices L and R for left and right, respectively, label the two cavities in contrast to indices $i\in\{1,\dots,M\}$ from the previous sections. The charge type operator becomes
\begin{align}
	&N=a^*_La^{\phantom{*}}_L+ a^*_R a^{\phantom{*}}_R + \tfrac{1}{2} \left( \sigma^z_L+ \id +\sigma^z_R +\id \right)=N_L+N_R \qquad \text{ where }\\
	& N_L=a^*_La^{\phantom{*}}_L+\tfrac{1}{2} \left( \sigma^z_L+ \id \right) \quad \text{and }\quad N_R= a^*_R a^{\phantom{*}}_R +\tfrac{1}{2} \left(\sigma^z_R +\id \right).
\end{align} Let $P^n$ denote the eigenprojections of $N$ as before and introduce $Q^{\mu}$ and $S^{\mu}$ for the eigenprojections (corresponding to the eigenvalue $\mu\in \mathbb{N}_0$) of $N_L$ and $N_R$, respectively. 
Introduce
\begin{align}
	\G:&=\mathcal{G}(\DriftJ, \sigma^z_L, \sigma^z_R, \HHtwo, \id )\\
	\mathfrak{l}:&=\Lie{ i\DriftJ, i\sigma^z_L, i\sigma^z_R, i\HHtwo, i \,\id }
\end{align}throughout this section where $\DriftJ:= \sum_{i=L,R} \left[ \oP{i} a^*_ia^{\phantom{*}}_i +\oA{i} \sigma^z_i+ \oI{i} (a^*_i \sigma^-_i + a^{\phantom{*}}_i \sigma^+_i )\right]$ denotes the drift and $\HHtwo:=\HH{\{(L,R)\}}=a^*_La^{\phantom{*}}_R+ a^{\phantom{*}}_La^*_R$ the hopping, neglecting the dependence on the set $I$ in the notation. Define $\Hil:=(\text{L}^2 (\mathbb{R} ) \otimes \mathbb{C}^2 )^{\otimes 2}$ to be the two cavity Hilbert space.
\newline

As Theorem~\ref{thm:sym_strategy} gives a list of sufficient conditions for strong operator controllability let us restate these conditions for the quantum control system $(\DriftJ,\sigma^z_L, \sigma^z_R, \HHtwo, \id,\sigma^x_L,\mathcal{H}, \hat{\mathcal{A}}) $:
	\begin{itemize}
		\item Self-adjointness: for every $H\in \{\sigma^z_L, \sigma^z_R, \HHtwo, \id,\sigma^x_L\}$, the operator $\DriftJ+H$ is essentially self-adjoint on $\dom (\DriftJ)\cap \dom(H)$ \textcolor{b1}{(Lemma~\ref{lem:selfad} for$M\in \mathbb{N}$)}
		\item $N$ is a charge-type operator, \textcolor{b1}{(proven in Lemma~\ref{lem:N_chargetype} for general $M\in \mathbb{N}$)}
		\item Recurrence: the control system satisfies the recurrence property, \textcolor{b1}{(Theorem~\ref{thm:rec_JCH} for $M\in \mathbb{N}$)}
		\item Symmetry: $\DriftJ, \sigma^z_L, \sigma^z_R, \HHtwo,$ and $\id$ commute with $N$, \textcolor{b1}{(Lemma~\ref{lem:commutation} for $M\in \mathbb{N}$)}
		\item Adapted Lie algebra rank condition: the dynamical group generated by $\DriftJ, \sigma^z_L, \sigma^z_R, \HHtwo,$ and $\id$ contains $\SU(N)$.
		\item Breaking the Symmetry: $\sigma_L^x$ is complementary to $N$.
		\item $\dim \Hil^{(n)}\ge2$ for at least one $n$.
	\end{itemize}
We added in \textcolor{b1}{blue} which of them were already checked in the previous section for a general number of cavities $M\ge 2$.
The following Lemma considers the multiplicities of eigenvalues of $N$, $N_R$ and $N_L$ and includes a proof of the last condition.
\begin{lem}[Eigenvalues of $N$ and $N_L$]\label{lem:2cav_Hn}
	$\phantom{1}$
	\begin{itemize}
		\item[(i)] Eigenvalues $n\ge1$ of $N$ have multiplicity $4n$, i.e., $d_n=\dim \Hil^{(n)}=4n$. Hence $\dim \Hil^{(n)}\ge 2 $ for $n\ge 1$.
		\item[(ii)] Eigenvalues $n\ge n_L\ge1$ of $N_L^{(n)}$ have multiplicity 4, eigenvalues $n_L=0$ or $n_L=n$ have multiplicity 2. We can decompose $\Hil^{(n)}=\bigoplus_{n_L=0}^n Q^{n_L} \Hil^{(n)}$ where we use the upper index notation from Eq.~\eqref{eq:notation_decomp_hil}. The same holds when replacing $L$ with $R$.
	\end{itemize}
\end{lem}
\begin{proof}
	(i) $\Hil^{(n)}=\text{span} \{\ket{n;n_L,b_L,b_R}| n,n_L\in \mathbb{N}_0, b_L,b_R\in \{0,1\}, n_L+b_R\le n, b_L\le n_L \}$  by direct calculation. Hence $\dim \Hil^{(n)}= 4(n+1)-4=4n$ if $n\ge 1 $ and $\dim \Hil^{(0)}=1$. The multiplicity of the eigenvalue $n$ corresponds to the dimension of the eigenspace $\Hil^{(n)}$. 
	
	(ii) Since $N_L$ commutes with $N$ we can perform a second direct sum decomposition of the Hilbert space where $Q^{n_L} \Hil^{(n)}=\text{span} \{ \ket{n;n_L;b_L,b_R}| b_L,b_R\in \{ 0,1\}, \, n_L+b_R\le n, \, b_L\le n_L \}$. This yields $\dim Q^{n_L}\Hil^{(n)} = 4$ if $1\le n_L \le n-1$ and $\dim Q^{n_L} \Hil^{(n)} = 2=\dim Q^{n_L}\Hil^{(n)}$ and is illustrated in Fig.~\ref{fig:blockdiag0} (B).
\end{proof}

\begin{figure}[ht]
	\centering
	\def\svgwidth{14cm}
	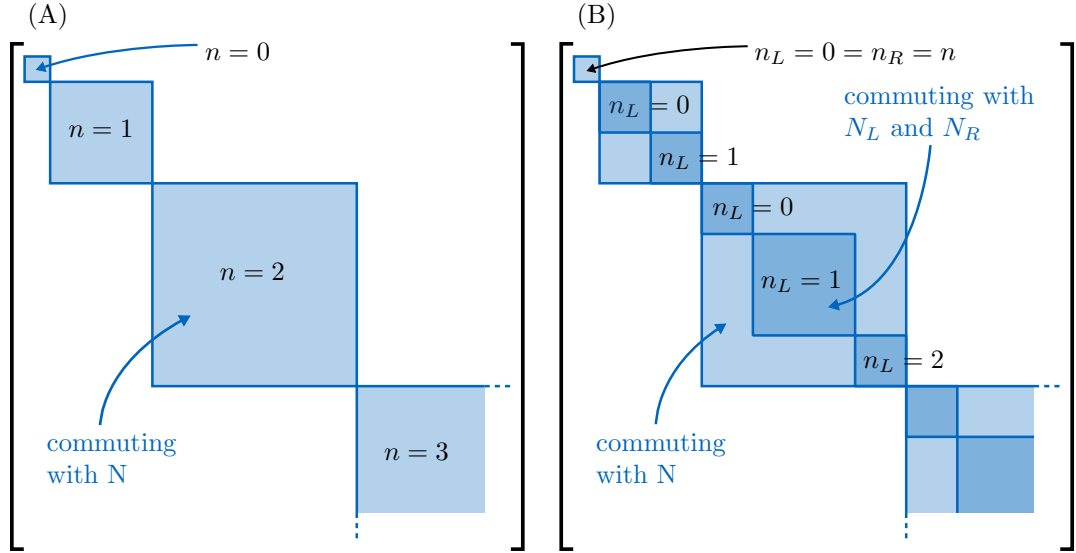
	\caption{(A) Block-diagonal structure of operators commuting with the total excitation operator $N$ is depicted in light blue shading; The size of the blocks corresponding to the eigenvalues $n$ of $N$ grows lienarly (as it is 1 for $n=0$ and $4n$ for $n\ge 1$).
		(B) Operators commuting with the excitation operator in the left and right cavity, $N_L$ and $N_R$, respectively admit yet another block-diagonal structure (darker blue shading). The size of these blocks is 2 or 4 (c.f~Lemma~\ref{lem:2cav_Hn}).}
	\label{fig:blockdiag0}
\end{figure}

In Section~\ref{sec:2cavities_rank} we show the `adapted Lie algebra rank condition', i.e., that $ \mathcal{SU}(N) \subseteq \mathcal{G}$. Proposition~\ref{prop:suff_trunc_compl} states that it is sufficient to show  
\begin{equation}\label{eq:lie_algebra_rank_K_2cav}
	\bigoplus_{n=0}^{K}\mathfrak{sl}(d_n; \mathbb{C}) \subseteq \lCK 
\end{equation}
for every $K\in \mathbb{N}_0$ where $d_n=4n$ for $n\ge 1$ and $d_0=1$ (Lemma~\ref{lem:2cav_Hn}) and where
\begin{equation}\label{eq:lCK_lCn_definition}
	\lCK := \LieC{\DriftJ^{[K]}, {\sigma^z_L}^{[K]}, {\sigma^z_R}^{[K]}, \HHtwo^{[K]}, \id^{[K]} }
\end{equation} with the notation from Eq.~\eqref{eq:notation_decomp_op}. 
We prove inclusion~\eqref{eq:lie_algebra_rank_K_2cav} by constructing generators of the l.h.s.~from elements of the r.h.s. In order to do so we exploit the following two symmetries of the system: firstly, all generators of $\mathfrak{l}$ commute with $N$ and secondly, the control Hamiltonians $\id$, $\DriftJ$, $\sigma^z_L$ and $\sigma^z_R$ admit yet another symmetry; they commute with $N_L$ and $N_R$. The block diagonal structure of such operators is illustrated in Fig.~\ref{fig:blockdiag0}.

\begin{dfn}[Operators $A_1,\dots,A_6$ and their projections]\label{dfn:JCH_2_A_i}
	Define the operators
	\begin{align}
		A_1&:= a^*_L \sigma^-_L  & A_2:&=a^{\phantom{*}}_L \sigma^+_L \nonumber\\
		A_3&:= a^*_R \sigma^-_R  & A_4:&=a^{\phantom{*}}_R \sigma^+_R \label{eq:JCH_2_Ai}\\
		A_5&:= a^*_L a^{\phantom{*}}_R  & A_6:&=a^{\phantom{*}}_L a^*_R \; .\nonumber
	\end{align}
	Let $Q^{n_L}$ and $S^{n_R}$ be the eigenprojections of $N_L$ and $N_R$, respectively, and introduce the notation 
	\begin{equation}
		A_j^{(n,n_L)}:=Q^{n_L} A_j^{(n)} Q^{n_L}= S^{n-n_L} A_j^{(n)} S^{n-n_L}
	\end{equation}
	for $j\in \{1,\dots,4\}$ and $0\le n_L\le n\in \mathbb{N}_0$.
	A third index denoting the excitation number $n_R$ of the second cavity can be omitted since it is fixed by the relation $n_L+n_R=n$ once the other two excitation numbers are given. Note that 
	\begin{equation}
		A_2^{(n,n_L)}=\left(A_1^{(n,n_L)}\right)^\text{T}\;, \quad A_4^{(n,n_L)}= \left(A_3^{(n,n_L)} \right)^\text{T}\;, \quad A_6^{(n)}=\left(A_5^{(n)}\right)^\text{T}\ .
	\end{equation}
\end{dfn}
Note that $A_1$ to $A_4$ admit both symmetries whereas $A_5$ and $A_6$ only commute with $N$. 
The proof idea for inclusion~\eqref{eq:lie_algebra_rank_K_2cav} is the following: First, we show in Lemma~\ref{lemma_technical_2} that certain projections of $A_1,\dots,A_6$ are in $\lCn$ and $\lCK$, respectively. Second, we restrict to one block for fixed $n$: Proposition~\ref{prop_ind_algebra_n} uses the above elements in $\lCn$ to construct generators of $\mathfrak{sl}(d_n; \mathbb{C})$. In Proposition~\ref{prop_algebra_K} we recursively extend this so that elements of $\lCK$ generate $\bigoplus_{n=0}^{K} \mathfrak{sl}(d_n; \mathbb{C})$. 

In Section~\ref{sec:2cavities_breaking} we show that the operator $\sigma^x_L$ is a complementary operator.

\subsection{Lie algebra of symmetry admitting operators}\label{sec:2cavities_rank}
This part proves inclusion~\eqref{eq:lie_algebra_rank_K_2cav}.

\begin{lem}[Elements of truncated Lie algebras]\label{lemma_technical_2}
	Let $A_1,\dots, A_6$, $Q^{n_L}$ and $S^{n_R}$ be as in Definition~\ref{dfn:JCH_2_A_i}.
	For natural numbers $n_L,n_R\le n\le K$ we get
	\begin{align}
		Q^{n_L}A_i^{(n)} Q^{n_L}\;, \quad S^{n_R}A_j^{(n)} S^{n_R}\; , \quad A_k^{(n)} &\;\in \mathfrak{l}^{(n)}_\mathbb{C} \label{eq:tec_incl_n} \\
		Q^{n_L} A_i^{[K]}Q^{n_L}\;, \quad S^{n_R}A_j^{[K]} S^{n_R}\;, \quad A_k^{[K]} &\; \in \mathfrak{l}^{[K]}_\mathbb{C} \label{eq:tec_incl_K}
	\end{align} where $i\in \{1,2 \},j \in \{3,4\}, k \in \{5,6\}$. We also have
	\begin{equation}\label{eq:technical_lemma_a*a}
		(a^*_ma^{\phantom{*}}_m)^{[K]}\in \LieC{ \id^{[K]}, (\sigma^z_m)^{[K]}}
	\end{equation} where $m\in\{L,R\}$.
\end{lem}

In order to prove this, we use the following result on the JC model.
\begin{lem}[Result on JC model; Corollary to Proposition 4.8 of \cite{keyl2014controlling}]\label{lem:JC_result}
	Define by $\mathfrak{m}^{[K']}$ the Lie algebra generated by the operators 
	\begin{equation}\label{eq:JC_generators}
		 i\sum_{\mu=0}^{K'}  Q_{\mu} (a^*_L \sigma^-_L  +a^{\phantom{*}}_L \sigma^+_L) Q_{\mu} \; , \; i \sum_{\mu=0}^{K'}Q_{\mu}\, \sigma^z_L Q_{\mu}
	\end{equation}
	and let $H$ be an anti-self-adjoint operator on $\operatorname{L}^2(\mathbb{R})\otimes\mathbb{C}^2$ such that $H\otimes \id\in\su(N_L) $. Then, for all $\mu\le K'$ we have that $Q_\mu (H\otimes \id)Q_\mu\in \mathfrak{m}^{[K']}$.
\end{lem}
\begin{proof}
	The statement follows from Proposition~4.8 of~\cite{keyl2014controlling} which states: the operators in Eq.~\eqref{eq:JC_generators} generate the algebra of anti-self-adjoint operators on $\Hil$ that are traceless on the eigenspaces of $N_L$.
\end{proof}

\begin{proof}[Proof of Lemma~\ref{lemma_technical_2}]
	Calculating $A_1-A_2= (2\oI{L})^{-1} \left[\DriftJ, \sigma^z_L \right]=\tfrac{1}{2}\left[a^*_L \sigma^-_L + a^{\phantom{*}}_L \sigma^+_L, \sigma^z_L\right]$ and $A_1+A_2= (4\oI{L})^{-1} \big[[\DriftJ,  \sigma^z_L], \sigma^z_L \big]$ shows that the operators $A_1$ and $A_2$ are elements of the Lie algebra $\lC$. Besides, they can be written as 
	\begin{align}
		A_1=B_1\otimes \id \;,\qquad & A_2=B_2\otimes \id ,
	\end{align}
	for the Jaynes-Cummings operators $B_1=a^* \otimes \sigma^-$ and $B_2=a \otimes \sigma^+$ acting on $\operatorname{L}^2(\mathbb{R})\otimes\mathbb{C}^2$. We apply Lemma~\ref{lem:JC_result} to the operators $i(B_1+B_2)$ and $B_1-B_2$. This gives 
	\begin{align*}\label{eq:incl_PAP0}
		Q_{n_L} ( C\otimes \id) Q_{n_L} \in \Lie{ i\sum_{\mu=0}^{K'}  Q_{\mu} (a^*_L \sigma^-_L  +a^{\phantom{*}}_L \sigma^+_L) Q_{\mu}, i \sum_{\mu=0}^{K'}Q_{\mu}\, \sigma^z_L Q_{\mu}  }
	\end{align*} 
	for all  $n_L\le K'$ and for $C\in \{(B_1-B_2), i(B_1+B_2)\}$.
	On the level of complexified Lie algebras we find the following: for all $n_L\le K'$ and $j\in \{1,2\}$ 
	\begin{equation}\label{eq:incl_PAP1}
		Q^{n_L} A_j Q^{n_L}  \in \LieC{ \sum_{\mu=0}^{K'} Q^{\mu} \DriftJ Q^{\mu}, \sum_{\mu=0}^{K'} Q^{\mu} \sigma^z_L Q^{\mu} } .
	\end{equation} 
	The generators of the Lie algebra on the right hand side are block diagonal with respect to $N$, $N_L$, and $N_R$ and act non-trivially only on the left cavity. For such a block diagonal operator $A$ we find that $\sum_{\mu=0}^{K'}Q^{\mu} A^{(K')} Q^{\mu}= A^{(K')}$. Hence we can perform two manipulations of inclusion~\eqref{eq:incl_PAP1}. On the one hand we project to $P^{K'}$ on both sides and set $K'=n$ to obtain
	\begin{equation}\label{eq:incl_PAP2}
		Q^{n_L} A_j^{(n)} Q^{n_L} \in \LieC{  \DriftJ^{(n)},(\sigma^z_L)^{(n)} } \subset \lCn
	\end{equation} for all $n_L\le n$. On the other hand, we can project to $P^{K'}$ on both sides of inclusion~\eqref{eq:incl_PAP1} and sum from $K'=0$ to $K'=K$. The same argument as before (block diagonal operators) allows for pulling this sum into the Lie algebra: For all $n_L\le K$ we find
	\begin{equation}\label{eq:incl_PAP3}
		Q^{n_L} A_j^{[K]} Q^{n_L} \in  \LieC{ \DriftJ^{[K]}, (\sigma^z_L)^{[K]} }\subset \lCK . 
	\end{equation} 
	For $A_3$ and $A_4$, the proofs are completely analogous, using $A_3=\id \otimes B_1$, $A_4=\id\otimes B_2 $ and replacing $L$ with $R$. 
	
	In order to prove the inclusions for $A_5$ and $A_6$, consider 
	\begin{align}\label{eq:JCH_2_proofAi_1}
		A_5 =\tfrac{1}{2}\left( \HHtwo + [a^*_La^{\phantom{*}}_L, \HHtwo]\right)\; , \quad A_6 =\tfrac{1}{2}\left(\HHtwo- [a^*_La^{\phantom{*}}_L, \HHtwo]\right).
	\end{align}
	This shows that $A_5$ and $A_6$ are elements of the Lie algebra generated by $a^*_La^{\phantom{*}}_L$ and $\HHtwo$. But since the photonic Hamiltonian $a^*_La^{\phantom{*}}_L$ is not among the control Hamiltonians, we will first show that on the level of truncated Hilbert spaces, we can construct it from operators commuting with $N_L$ and $N_R$. We know from the Jaynes-Cummings model that $P^{n_L}_\JC i(a^*a\otimes \id) P^{n_L}_\JC$ can be written as a linear combination of $iP^{n_L}_\JC$ and $P^{n_L}_\JC i(\id\otimes \sigma^z)P^{n_L}_\JC$ for every $n_L\in\mathbb{N}$ (c.f~Remark after Proposition 4.8 from~\cite{keyl2014controlling}). This can be again converted to a statement on complexified Lie algebras on the JCH model; hence
	\begin{equation}\label{eq:JCH_2_proofAi_2}
		Q^{n_L} a^*_La^{\phantom{*}}_L Q^{n_L}\in \LieC{ Q^{n_L} ,  Q^{n_L} \sigma^z_LQ^{n_L} }\; .
	\end{equation} 
	The generators of the Lie algebra and $a^*_La^{\phantom{*}}_L$ are again block diagonal with respect to $N$, $N_L$, and $N_R$. Hence the same considerations as for $A_1$ and $A_2$ apply. We project both sides of inclusion~\eqref{eq:JCH_2_proofAi_2} to $P_n$ and sum from $n_L=0$ to $n_L=n$. This or additionally summing from $n=0$ to $n=K$ yield  
	\begin{align}
		(a^*_La^{\phantom{*}}_L)^{(n)} \in \LieC{\id^{n}, (\sigma^z_L)^{(n)} } \quad \text{or} \quad (a^*_La^{\phantom{*}}_L)^{[K]}\in \LieC{ \id^{[K]}, (\sigma^z_L)^{[K]}}\; ,
	\end{align} respectively. Combining this with Eq.~\eqref{eq:JCH_2_proofAi_1} gives the desired inclusions $ A^{(n)}_5 \in \lCn$ and $A_5^{[K]} \in \lCK$.\\
	We obtain the inclusions for $A_6$ by replacing $L$ with $R$.
\end{proof}

The operators from Lemma~\ref{lemma_technical_2} are used as a starting point to construct generators of $\mathfrak{sl}(d_n; \mathbb{C})$ and $ \bigoplus_{n=0}^{K}\mathfrak{sl}(d_n; \mathbb{C})$. We start with $\mathfrak{sl}(d_n; \mathbb{C})$ in the following Proposition and hence fix the total excitation number $n$.

\begin{prop}[Lie algebra inclusion for fixed $n$]\label{prop_ind_algebra_n}
	For all $n\in \mathbb{N}_0$, we have $ \mathfrak{sl}(d_n,\mathbb{C}) \subseteq \lCn $.
\end{prop}
Remember that $ \mathfrak{sl}(d_n,\mathbb{C})$ denotes the special linear algebra of degree $d_n=4n$ for $n\ge1$ and $d_0=1$ over the complex numbers and that $\lCn $ is the complex Lie algebra generated by $\DriftJ^{(n)}$, ${\sigma^z_L}^{(n)}$, ${\sigma^z_R}^{(n)}$, $\HHtwo^{(n)}$, and $\id^{(n)}$. We make some preliminary considerations for the proof:
 \begin{itemize}
 	\item[1)] We introduce an ordering of the basis vectors $\ket{n;n_L;b_L,b_R}$ of $\Hil^{(n)}$ that is given by the following bijective map
 	\begin{equation}\label{eq:basis_ordering}
	 	\ket{n;n_L;b_L,b_R} \mapsto 4n_L+2b_L+b_R-1+2\delta_{n_L,0}-\delta_{n_L,n}\delta_{b_L,1}.
 	\end{equation}For fixed $n$, the basis is mapped to the set $\{1,2,\dots,4n\}$. 
 	Vectors are arranged by firstly fixing $n_L$ and then ordering $(b_L,b_R)$ as follows: $(0,0),(0,1),(1,0),(1,1)$.
 	\item[2)] Operators acting on the finite dimensional subspaces $\Hil^{(n)}$ can be represented as complex $(4n\times4n)$-matrices. Using the basis ordering from 1), introduce the matrix $E_{i,j}$ for $i,j\in \{1, \dots, 4n\}$ which is defined by its entries $k,l$ given by 
 	\begin{equation}
	 	(E_{i,j})_{k,l}= \delta_{i,k} \delta_{j,l}=
	 	\begin{cases}
		 	1, & \text{if } k=i \text{ and } l=j\\
		 	0, & \text{otherwise}.
	 	\end{cases}
 	\end{equation}
 	with the Kronecker delta $\delta$. 
 	This simplifies calculations with commutators as they yield
 	\begin{equation}
	 	[E_{i,j},E_{k,l}]=E_{i,l}\delta_{j,k}-E_{k,j}\delta_{l,i} \quad .
 	\end{equation}
 	Additionally, the above matrices can be used to define generators of $\mathfrak{sl} \left( 4n; \mathbb{C}\right)$. Take $E_{i,i+1}$ and $E_{i+1,i}=(E_{i,i+1})^\text{T}$ for $i\in\{1,\dots, 4n-1\}$. They generate $\mathfrak{sl} \left( 4n; \mathbb{C}\right)$ since one can construct arbitrary traceless $4n\times4n$-matrices by taking commutators of them.
 \end{itemize}
Hence the proof idea is to construct $E_{i,i+1}$ for $i\in\{1,\dots, 4n-1\}$ via mutual commutators and linear combinations of $A_1^{(n,n_L)},\dots, A_4^{(n,n_L)}$, $A_5^{(n)}$ and $A_6^{(n)}$.

\begin{figure}[ht]
	\centering
	\def\svgwidth{14cm}
	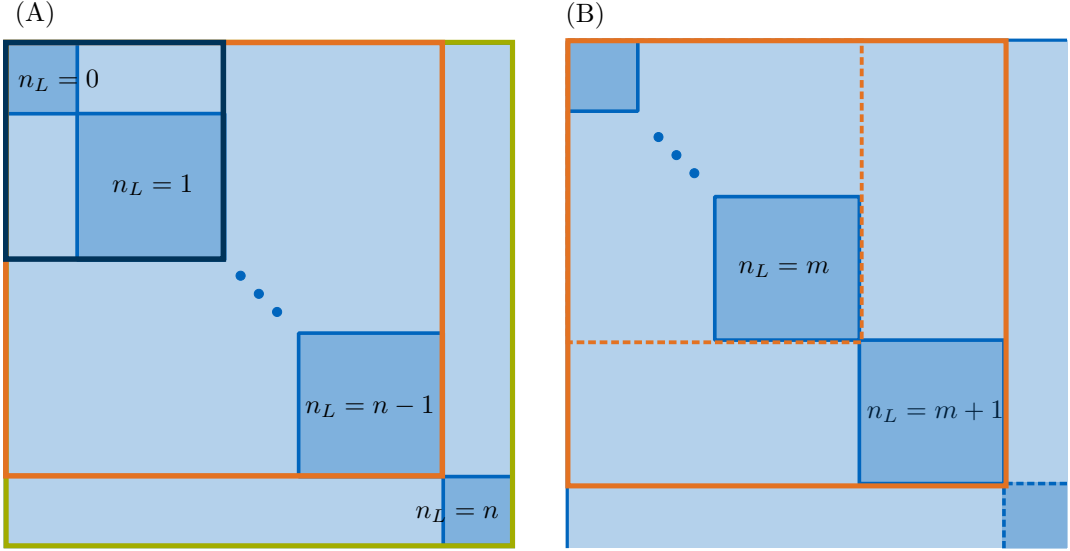
	\caption{(A) Block for fixed total excitation number $n\ge1$ with basis ordering as in~\eqref{eq:basis_ordering}; Blockdiagonal structure of operators commuting with $N$ in light blue and with $N_L$ and $N_R$ in darker blue shading. Illustration of the proof of Proposition~\ref{prop_ind_algebra_n}: Via multiple commutators of $A_i^{(n,n_L)}$ and $A^{(n)}_j$ where $i=1,\dots,4$ and $j=5,6$ we construct the matrices $E_{i,i+1}$. Note that $A_i^{(n,n_L)}$ for $i=1,\dots,4$ are block-diagonal w.r.t.~$N_L $ and $N_R$ and have non-zero entries only inside one of the $n_L$-blocks; in contrast, $A_5^{(n)}$ and $A_6^{(n)}$ only commute with $N$ and hence live in the rest of the light blue block for fixed $n$. The matrices $E_{i,i+1}$ have only one non-zero entry that is on the first offdiagonal and they generate the Lie algebra of all traceless matrices confined the light blue block with green edging, i.e., $\mathfrak{sl}(4n,\mathbb{C})$; First we construct those $E_{i,i+1}$ that are confined to the block with dark blue edging which serves as an induction basis; We inductively extend this to all $E_{i,i+1}$ in the orange edging and in a final step to the remaining ones in the green edging. 
	(B) Illustration of the inductive step where generators $E_{i,i+1}$ for $i\le 4m+1$ (inside block with dashed orange edging) are in $\mathfrak{sl}(4n,\mathbb{C})$ by hypothesis and where we construct those $E_{i,i+1}$ with $i=4m+2,\dots, 4m+5$.}
	\label{fig:blockdiag}
\end{figure}

\begin{proof}[Proof of Proposition~\ref{prop_ind_algebra_n}]
	For $n=0$, the statement is trivial since the set $\mathfrak{sl}(1;\mathbb{C})$ is empty.
	
	For $n\ge1$, we have $d_n=4n$. The statement to show is hence $\mathfrak{sl}(4n;\mathbb{C})\subseteq \lCn $. Due to Lemma~\ref{lemma_technical_2}, the operators  $A_1^{(n,n_L)},\dots,A_4^{(n,n_L)}$, $A_5^{(n)}$, and $A_6^{(n)}$ are elements of $\lCn$ for $n_L\le n$. We use them to construct the generators $E_{i,i+1}$ of $\mathfrak{sl}(4n;\mathbb{C})$ -- the matrices $E_{i+1,i}$ are constructed analogously by taking transposes. For $n=1$, straightforward calculation gives the desired result. 
	
	For $n\ge 2$, we prove by induction on $m$ that $\mathfrak{sl}(2+4m;\mathbb{C})\subseteq \lCn$ for all $1\le m \le n-1$. The proof idea is graphically illustrated in Fig.~\ref{fig:blockdiag}.
	
	The \textbf{induction basis} corresponds to the statement for $m=1$. The matrices
	\begin{align}
		A_1^{(n,1)}&=E_{3,5}+E_{4,6}&=\left(A_2^{(n,1)}\right)^\text{T}\\
		A_3^{(n,0)}&=\sqrt{n}E_{1,2} &=\left(A_4^{(n,0)}\right)^\text{T}\\
		A_3^{(n,1)}&=\sqrt{n-1}\left(E_{3,4}+E_{5,6}\right)&=\left(A_4^{(n,1)}\right)^\text{T}\\
		A_5^{(n)}&=\sqrt{n}  E_{1,3}+ \sqrt{n-1}E_{2,4} + F_5 &=\left(A_6^{(n)}\right)^\text{T}
	\end{align} are in $\lCn$. 
	Note that in contrast to $A_1^{(n,0)}$, $A_4^{(n,0)}$, $A_3^{(n,1)}$, and $A_4^{(n,1)}$, the operators $A_5^{(n)}$ and $A_6^{(n)}$ are not constrained to $i,j\le6$ (dark-blue edged block in Fig.~\ref{fig:blockdiag} (A)), but they have non-zero entries in the rest of $\Hil^{(n)}$; the matrix entries of $F_5$ satisfy $(F_5)_{i,j}= 0$ for $i,j\le2$ but may be non-zero for $i,j\ge3$. 
	Thus calculate the commutators
	\begin{align}
		\left[A_3^{(n,0)},A_5^{(n)} \right]&=\left[ \sqrt{n}E_{1,2}, \sqrt{n} E_{1,3}+ \sqrt{n-1} E_{2,4} \right]= \sqrt{n}\sqrt{n-1} E_{1,4}\nonumber\\
		\left[A_4^{(n,0)},A_5^{(n)} \right]&=\left[ \sqrt{n}E_{2,1},\sqrt{n}  E_{1,3}+ \sqrt{n-1}E_{2,4}\right]=\sqrt{n}^2 E_{2,3} \; . \label{eq:comm_A_3A_5} 
	\end{align}
	Hence $ E_{1,4}$, and $E_{2,3}$ as well as their transposes are elements of $ \lCn$. 
	Use them to calculate
	\begin{align*}
		\left[E_{1,4},A_4^{(n,1)}\right]&= \left[E_{1,4},\sqrt{n-1}\left(E_{4,3}+E_{6,5}\right) \right]= \sqrt{n-1} E_{1,3}
	\end{align*}
	and furthermore $\left[E_{3,1},E_{1,4}\right]=E_{3,4}$, and $\big[E_{4,3},A_1^{(n,1)}\big]=\left[E_{4,3}, E_{3,5}+E_{4,6}\right]= E_{4,5}$. 
	The last missing generators are $E_{5,6}$ and $E_{6,5}$. They are obtained via
	\begin{align*}
		\left[E_{5,4},A_1^{(n,1)}\right]+E_{3,4}&=\left[E_{5,4}, E_{3,5}+E_{4,6}\right]+E_{3,4}=E_{5,6} \;.
	\end{align*}
	For the \textbf{inductive step} assume that $\mathfrak{sl}(2+4m)\subseteq \lCn$ for some $1\le m \le n-2$ and show that $\mathfrak{sl}(2+4(m+1))\subseteq \lCn$. By induction hypothesis $E_{i,j}\in \lCn$ for $i,j\le 4m+2$. For suitable constants $\alpha_1,\dots, \alpha_3\in \mathbb{R}$ we have
	\begin{align*}
		&\left[ E_{4m+2,4m-1},A_5^{(n)} \right]+\alpha_1 E_{4m-2,4m-1}= \sqrt{m}\sqrt{n-m}E_{4m+2,4m+3}\\
		&\left[ E_{4m+2,4m},A_5^{(n)}\right]+\alpha_2 E_{4m-2,4m}= \sqrt{m-1}\sqrt{n-m}E_{4m+2,4m+4}\\
		&\left[ E_{4m+2,4m+1},A_5^{(n)}\right]+\alpha_3 E_{4m-2,4m+1}= \sqrt{m}\sqrt{n-m-1}E_{4m+2,4m+5}\; .
	\end{align*}
	Calculating commutators of these operators with each other yields for suitable $\alpha_4\in \mathbb{R}$:
	\begin{align*}
		&\left[E_{4m+3,4m+2},E_{4m+2,4m+4}\right]=E_{4m+3,4m+4}\\
		&\left[E_{4m+4,4m+2},E_{4m+2,4m+5}\right] = E_{4m+4,4m+5}\\
		&\left[E_{4m+5,4m+4},A_1^{(n,m+1)}\right]+\alpha_4 E_{4m+3,4m+4} \sim E_{4m+5,4m+6}\; .
	\end{align*}
	We just showed that $E_{i,i+1}$ for $i\in \{4m+2,\dots,4m+5\}$ and their transposes are elements of $ \lCn$. Hence $\mathfrak{sl}(2+4(m+1))\subseteq \lCn$ by linearity which concludes the induction. \\
	In order to finish the proof of Proposition~\ref{prop_ind_algebra_n}, we have to construct the missing generators $E_{4n-1,4n-2}$, $E_{4n-1,4n}$ (and their transposes) of $\mathfrak{sl}(4n;\mathbb{C})$. Due to the previous induction, matrices $E_{i,j}$ are elements of $\lCn$ for $i,j\le 4n-2$. Hence calculate
	\begin{align*}
		&\left[ E_{4n-2,4n-5}, A_5^{(n)} \right]+\alpha_5 \,E_{4n-6,4n-5}\sim E_{4n-2,4n-1}\\
		&\left[ E_{4n-2,4n-3}, A_5^{(n)} \right]+\alpha_6 E_{4n-6,4n-3}\sim E_{4n-2,4n}\\
		& \left[E_{4n-1,4n-2}, E_{4n-2,4n}\right]=E_{4n-1,4n}
	\end{align*}which concludes the proof.
\end{proof}

\begin{prop}[Adapted Lie algebra rank condition]\label{prop_algebra_K}
	For all $K\in \mathbb{N}_0$, we have $\bigoplus\limits_{n=0}^K \mathfrak{sl} (d_n;\mathbb{C})\subseteq \lCK$.
\end{prop} 

The Lie algebra $\lCK$ is defined in Eq.~\eqref{eq:lCK_lCn_definition}. The construction of generators of $\bigoplus_{n=0}^K \mathfrak{sl} (4n;\mathbb{C})$ differs slightly from Proposition~\ref{prop_ind_algebra_n}. In the proof of the former all excitation numbers ($n$, $n_R$, and $n_L$) of the operators $A_1,\dots A_4$ were fixed. However, now only one of them is, either $n_L$ or $n_R$. 
If one considers the commutator of an operator with fixed $n_L$ and one with fixed $n_R$, then the result is no longer confined to any excitation number. We avoid the above-mentioned problems by exploiting the relation $n_L+n_R=n$ which ensures that $n_R\le n$.

\begin{proof}[Proof of Proposition~\ref{prop_algebra_K}]
	The statement is trivial for $K=0$ since $\mathfrak{sl}(1;\mathbb{C})$ is empty.
	
	For $K\ge 1$ the proof idea is to start on $\mathfrak{sl} (4K;\mathbb{C})$ and recursively construct the other generators of $\mathfrak{sl} (4n;\mathbb{C})$ where $n\le K$. For $0\le n \le K$, the matrices $\{0\}\oplus E_{l,l\pm1}^{(n)}\oplus \{0\}$ with $1\le l\le 4n$ are generators of $\bigoplus_{n=0}^K \mathfrak{sl} (4n;\mathbb{C})$. 
	They are constructed from $\lCK$ as in Proposition~\ref{prop_ind_algebra_n} with one difference: We replace $A_i^{(n)}$ by $A_i^{[K]}$ for $i=1,\dots,6$. Lemma~\ref{lemma_technical_2} gives that for $n_L,n_R\le K$
	\begin{equation}
		Q^{n_L} A_i^{[K]} Q^{n_L},\, S^{n_R}A_j^{[K]} S^{n_R},\, A_k^{[K]} \in \lCK 
	\end{equation} where $i=1,2$, $j=3,4$, and $k=5,6$.
	
	Let us start by setting $n_R=K$, i.e.,~by considering the operators $S^{n_R=K} A^{[K]}_jS^{n_R=K}\in \lCK$ for $j=3,4$. They act non-trivially only on the subspace $\Hil^{(n=K)}$, since the restrictions $n_R\le n\le K$ and $n_R=K$ imply that $n=K$. Now, proceed analogously to the proof of Proposition~\ref{prop_ind_algebra_n} beginning with the calculations from Eq.~\eqref{eq:comm_A_3A_5}, but replace $A_i^{(n)}$ by $A_i^{[K]}$. We find
	\begin{align}
		\left[S^{K} A^{[K]}_3 S^{K},A_5^{[K]} \right]&\sim E^{(K)}_{1,4}\; , & \left[S^{K} A^{[K]}_4 S^{K},A_6^{[K]} \right]&\sim E^{(K)}_{4,1}\; , \nonumber\\
		\left[S^{K} A^{[K]}_4 S^{K},A_5^{[K]} \right]&\sim E^{(K)}_{2,3}\;, & \left[S^{K} A^{[K]}_3 S^{K},A_6^{[K]} \right]&\sim E^{(K)}_{3,2} \; . \label{eq:algebra_K_firstgenerator}
	\end{align}
	We thereby constructed the first generators on $\Hil^{(n=K)}$. 
	Following the proof of Proposition~\ref{prop_ind_algebra_n} the other generators of $\mathfrak{sl}(4K;\mathbb{C})$ can be traced back to concatenated commutators of elements in $\lCK$ with the results of Eq.~\eqref{eq:algebra_K_firstgenerator}. Hence $\{0\}\oplus \mathfrak{sl}(4K; \mathbb{C})\subseteq \lCK$ by linearity.
	
	Now proceed with $n_R=K-1$: In Eq.~\eqref{eq:algebra_K_firstgenerator} replace $S^{K} A^{[K]}_3 S^{K}$ by $S^{K-1} A^{[K]}_3 S^{K-1}$. The results are linear combinations of (one of) $E^{(K)}_{1,4}$, $E^{(K)}_{3,2}$, $ E^{(K)}_{2,3}$, or $E^{(K)}_{4,1}$ and (one of) $E^{(K-1)}_{1,4}$, $E^{(K-1)}_{3,2}$, $E^{(K-1)}_{2,3}$, or $E^{(K-1)}_{4,1}$. But since the former are already in $\lCK$, we can simply subtract them from the result and hence the latter also are. Use these generators of $\{0\}\oplus \mathfrak{sl}(4(K-1); \mathbb{C})\oplus \{0\}$ again as a starting point to construct all others. 
	Recursively apply this procedure until $n_R=0$ to construct all generators of $\bigoplus_{n=0}^K \mathfrak{sl} (d_n;\mathbb{C})$ and thereby conclude the proof.
\end{proof}

\subsection{Breaking the symmetry}\label{sec:2cavities_breaking}
What still needs to be found is the complementary operator. Note that this operator has to break the U(1) symmetry, i.e., it must not commute with $N$. 
A natural candidate is the atom operator $\sigma^x_L$, since it interchanges the atom's excited and ground state in cavity $L$.
Let us define the subspaces $\mathcal{H}_0:=\{0\}$  and
\begin{align*}
	\mathcal{H}_+ &:= \text{span} \left\{ \ket{n;n_L;1,b_R}\in\mathcal{H} | n,n_L \in \mathbb{N}, b_R\in \{0,1\}, n_L+b_R\le n \right\} ,\\
	\mathcal{H}_- &:= \text{span} \left\{ \ket{n;n_L;0,b_R}\in\mathcal{H} |  n,n_L \in \mathbb{N}_0, b_R\in \{0,1\}, n_L+b_R\le n\right\} ,
\end{align*} such that $\mathcal{H} = \mathcal{H}_- \oplus \mathcal{H}_0 \oplus \mathcal{H}_+$. Define $E_\pm$ as projections onto $\Hil_\pm$ and $\mathcal{H}^{(n)}_\pm:=P^{(n)}\mathcal{H}_\pm$.

\begin{prop}[Complementarity] \label{prop:complementary}
	The operator $\sigma^x_L$ is complementary to $N$.
\end{prop}
\begin{proof}
	 We check $(i)$--$(vi)$ from Definition~\ref{dfn:compl} for the decomposition $\mathcal{H} = \mathcal{H}_- \oplus \mathcal{H}_0 \oplus \mathcal{H}_+$.	 It is easy to see that $E_\pm$ commutes with $P_n$ and that $P_0 E_-=P_0$. Defining $P^{(n)}_{\pm}=P_n E_\pm$ we find that they are non-zero for all $n \in \mathbb{N}$. Condition $(i)$ of Definition~\ref{dfn:compl} is trivial because the subspace $\mathcal{H}_0$ is empty. 
	 Since $\sigma^x_L$ is a bounded operator we have $ \dom(\sigma^x_L)=\mathcal{H}$ and the first statement of $(ii)$ is also true.
	 To prove that $\sigma^x_L P_+^{(n+1)}\ket{\psi}= P_-^{(n)}\sigma^x_L \ket{\psi}$ holds for any $\ket{\psi} \in \mathcal{H}$ and $n>1$, consider the action of $\sigma^x_L$ on basis vectors of $\mathcal{H}$:
	 \begin{align*}
		 & \sigma^x_L \ket{n;n_L;b_L,b_R}= \ket{n-2b_L+1;n_L-2b_L+1;1-b_L,b_R}.
	 \end{align*}
	 The operator $\sigma^x_L$ changes the total excitation number by lowering or raising the atomic excitation number of the left cavity if $b_L=1$ or $b_L=0$, respectively. The short calculations 
	 \begin{align*}
		 P_-^{(m)} \sigma^x_L & \ket{n;n_L;b_L,b_R}= 
		 \begin{cases}
			 \ket{n-1;n_L-1;0,b_R}, & \text{if }b_L=1 \text{ and } n-1=m\\
			 0, & \text{otherwise, }
		 \end{cases}\\
		 \sigma^x_L P_+^{(m+1)}& \ket{n;n_L;b_L,b_R}= 
		 \begin{cases}
			 \ket{n-1;n_L-1;0,b_R}, & \text{if }b_L=1 \text{ and } n-1=m\\
			 0, & \text{otherwise}
		 \end{cases}
	 \end{align*} 
	 yield that $P_-^{(m)} \sigma^x_L \ket{n;n_L;b_L,b_R}= \sigma^x_L P_+^{(m+1)} \ket{n;n_L;b_L,b_R}$. By linearity, $P_-^{(n)} \sigma^x_L \ket{\psi}= \sigma^x_L P_+^{(n+1)} \ket{\psi}$ for all $\ket{\psi} \in \mathcal{H}$ and $n\ge 1$. Since the operator $P_-^{(m)} \sigma^x_L P_+^{(m+1)} $ is bounded and we have 
	 \begin{equation*}
		 P_-^{(m)} \sigma^x_L P_+^{(m+1)} \ket{n;n_L;b_L,b_R}= 
		 \begin{cases}
			 \ket{n;n_L;b_L,b_R}, & \text{if }b_L=1 \text{ and } n=m+1\\
			 0, & \text{otherwise},
		 \end{cases}
	 \end{equation*}
	 its restriction to the orthogonal complement of the kernel of $P_-^{(m)} \sigma^x_L P_+^{(m+1)} $  
	 is an isometry. Thus, it is a partial isometry with $P_+^{(m+1)}$ as its source and $P_-^{(m)}$ as its target projection.
	 Condition $(iii)$ is satisfied; this can be easily seen by examining the subspaces $\mathcal{H}_{\phantom{-}}^{(0)}= \text{span} \left\{ \ket{0;0;0,0}\right\}$ and $\mathcal{H}_+^{(1)}=\text{span} \left\{ \ket{1;1;0,1}\right\}$. They are one-dimensional and thus, there are no traceless unitaries in $\SU(N)$ that commute with $P_{\phantom{-}}^{(0)}\oplus P_+^{(1)}$. The space of one-dimensional projections onto $\mathcal{H}_{\phantom{-}}^{(0)}$ and $\mathcal{H}_+^{(1)}$ is simply spanned by $P_{\phantom{-}}^{(0)}$ and $P_+^{(1)}$ which are exactly those projections the group generated by $\exp(it\sigma^x_L)$ acts transitively on.	
\end{proof}

\begin{figure}[ht]
	\centering
	\def\svgwidth{12cm}
	\large
	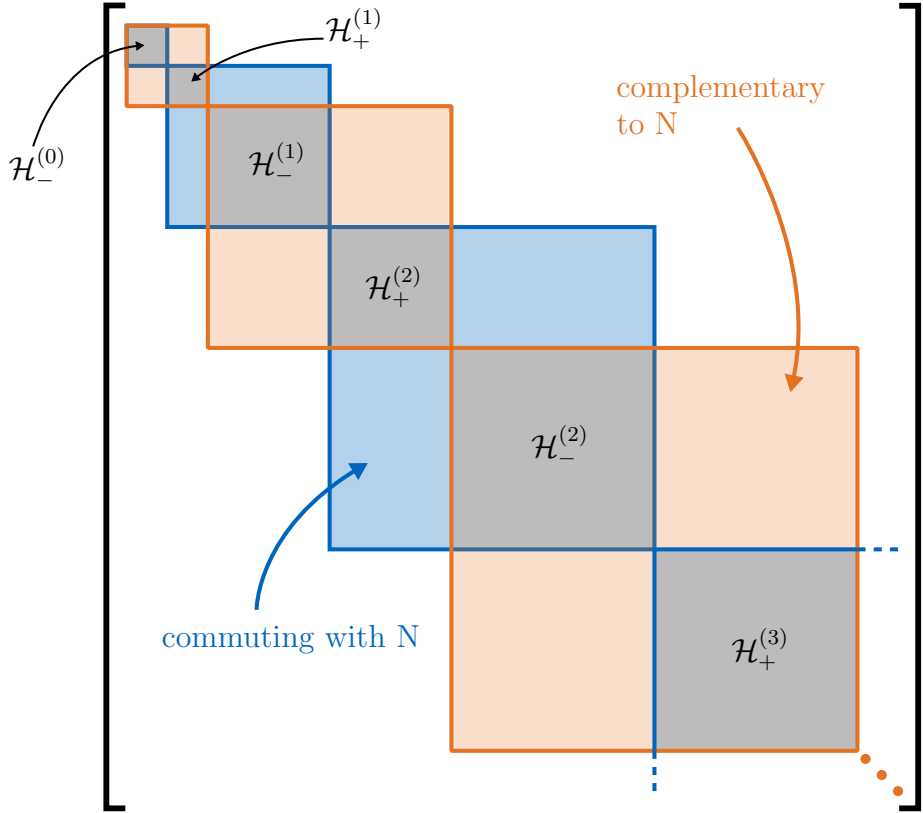
	\caption{Block-diagonal structure of operators commuting with $N$ (light blue) and operators complementary to $N$ (orange); operators that only act nontrivially on $\Hil^{(n)}_\pm$ are those confined to overlapping blocks (light blue and orange); complementary operator $\sigma^x_L$ interchanges $\mathcal{H}^{(n)}_-$ and $\mathcal{H}^{(n+1)}_+$. }
	\label{fig:blockdiag_compl}
\end{figure}

The blockdiagonal structure of operators commuting with $N$ and those complementary to $N$ is illustrated in Fig.~\ref{fig:blockdiag_compl}.

\subsection{Controllability result}
Combining the results of this section, we conclude controllability of the JCH model for two cavities.
\begin{thm}[Controllability of the JCH control system for two cavities]\label{thm:control_2}
	For $M=2$, the control system from Proposition~\ref{prop:JCH_con_sys} is strongly operator controllable.
\end{thm}
\begin{proof}
	The recurrence condition is satisfied due to Theorem~\ref{thm:rec_JCH}. The drift $\DriftJ$ and every control Hamiltonian $H\in\{\id$, $\sigma^z_L$, $\sigma^z_R$, $\sigma^x_L$, $\HHtwo\}$ admit $\fset$ as a total set of analytic vectors and generate a Lie algebra such that all its elements, e.g.~all $\DriftJ+H$, have unique self-adjoint extensions (Lemma~\ref{lem:selfad}). There exists a charge type operator $N$ (Lemma~\ref{lem:N_chargetype}) that commutes with the operators $\id$, $\sigma^z_L$, $\sigma^z_R$, $\HHtwo$ and $\DriftJ$ (Lemma~\ref{lem:commutation}). Combining Proposition~\ref{prop:suff_trunc_compl} with Proposition~\ref{prop_algebra_K}, we find that these operators generate a dynamical group which contains the Lie group $\SU (N)$. The control Hamiltonian $\sigma^x_L$ is complementary to $N$ (Proposition~\ref{prop:complementary}). Additionally, we find that $\dim \Hil^{(n)}=4n> 2$ for $n\ge 1$ (Lemma~\ref{lem:2cav_Hn}). Hence all conditions in Theorem~\ref{thm:sym_strategy} are satisfied and the quantum control system for $M=2$ is strongly operator controllable.
\end{proof}

\section{Controllability of the general JCH model}\label{sec:graph}

In this chapter, we inductively generalize the result from Theorem~\ref{thm:control_2} to an arbitrary number of cavities. We provide a proof for Theorem~\ref{thm:contr_JCH}, i.e., we show strong operator controllability of the Jaynes-Cummings-Hubbard model.

\subsection{Graph representation of JCH model}
It will be convenient to represent the JCH model with $M$ cavities as a graph with $M$ vertices. Every vertex of the graph stands for a cavity containing a harmonic oscillator mode and a two-level system. An edge represents photon hopping between the corresponding cavities. 
This graph is referred to as the hopping interaction graph of the quantum control system. It is given by a set $I\subseteq \{(i,j) |\,1\le i<j\le M\}$ such that an element $(i,j)\in I$ represents an edge between vertices $i$ and $j$. 
For simplicity and also for physical reasons, we only consider control systems that correspond to connected hopping interaction graphs.

The proof of controllability of the JCH model builds on this graph representation and is depicted in Fig.~\ref{fig:graph_ind}: We use an induction on the number of cavities where the basis is proven in the previous section. We divide the control system for $M+1$ cavities into a system for $M$ cavities and another one for two cavities such that both systems overlap on the $M$th cavity. The first one is controllable by induction hypothesis and the second one by inductions basis. We conclude that therefore the whole system also is. In order to do so, we use a general result on the controllability of overlapping quantum control systems that is stated in the following proposition.

\begin{prop}[Controllability of overlapping systems]\label{prop:overlapping_Keyl}
	Let $\mathcal{K}_1$, $\mathcal{K}_2$, and $\mathcal{K}_3$ be three separable, potentially infinite dimensional Hilbert spaces and consider self-adjoint operators $H_1,\dots, H_n$ on $ \mathcal{K}_1\otimes \mathcal{K}_2$ and $K_1,\dots, K_m$ on $\mathcal{K}_2\otimes \mathcal{K}_3$. Let the relations $\G(H_1,\dots, H_n)=\mathcal{U}(\mathcal{K}_1\otimes \mathcal{K}_2)$ and $\G(K_1,\dots, K_m)= \mathcal{U}(\mathcal{K}_2 \otimes \mathcal{K}_3)$ hold.
	Then 
	\begin{equation}
		\G(H_1\otimes \id,\dots, H_n\otimes\id, \id\otimes K_1, \dots, \id\otimes K_m)=\mathcal{U}(\mathcal{K}_1\otimes \mathcal{K}_2\otimes \mathcal{K}_3)\ .
	\end{equation}
\end{prop}
The proof is given in Theorem~4.5 of~\cite{heinze2016master} and in Lemma~5.5 of~\cite{hofmann2017controlling}.

\subsection{Proof of Theorem~\ref{thm:contr_JCH}}
We use following two Lemmas in the proof of Theorem~\ref{thm:contr_JCH} when we want to conclude equality between the two dynamical groups
\begin{align*}
	&\G (\id,\sigma^z_1, \dots, \sigma^z_{M+1} ,\sigma^x_1, \dots, \sigma^x_{M+1},\DriftJ, \HH{I} ) \ \text{ and }\\
	&\G (\id,\sigma^z_1, \dots, \sigma^z_{M+1} ,\sigma^x_1, \dots, \sigma^x_{M+1},\DriftJ', \DriftJ'',H'_\tH , H''_\tH )\ .
\end{align*} The first Lemma considers the different hopping Hamiltonians of both groups whereas the second one the different drifts.

\begin{lem}[Hopping Hamiltonians]\label{lem:suff_hop_JCH}
	Consider the quantum control system from Proposition \ref{prop:JCH_con_sys} for $M+1$ cavities with a tree-like hopping interaction graph $J$. Let it be strongly operator controllable. Let $I$ be a supergraph of $J$ such that $J$ is a spanning tree of $I$, set $J_M:=J \setminus \{(M,M+1)\}$ and $\HH{I}=\sum_{(i,j)\in I} (a^*_i a^{\phantom{*}}_j +a^{\phantom{*}}_i a^*_j)$, $H'_\tH:=\sum_{(i,j)\in J_M} (a^*_i a^{\phantom{*}}_j +a^{\phantom{*}}_i a^*_j)$ and $H''_\tH:= a^*_M a^{\phantom{*}}_{M+1} + a^{\phantom{*}}_{M} a^*_{M+1}$. 
	Then 
	\begin{equation}
	\exp \left( itH'_\tH\right),\;\exp \left( itH''_\tH \right) \in \G (\id,\sigma_1^z,\dots, \sigma^z_{M+1},\HH{I} )\; .
	\end{equation}
	and the system corresponding the hopping interaction graph $I$ is strongly operator controllable.
\end{lem}
\begin{proof}
	Consider the hopping control Hamiltonian $\HH{I}$ for $M+1$ cavities with a connected hopping interaction graph $I$ and let $J$ denote a spanning tree subgraph of $I$ such that $(M,M+1)\in J$ and $H'_\tH + H''_\tH=\sum_{(i,k)\in J} (a^*_i a^{\phantom{*}}_k +a^{\phantom{*}}_i a^*_k)$. We find for all $K\in \mathbb{N}_0$ and $(i,k)\in J\subseteq I$ that
	\begin{align}
	\left[\big(a^*_i a^{\phantom{*}}_i\big)^{[K]},\big[\HH{I}^{[K]}, \big(a^*_k a^{\phantom{*}}_k \big)^{[K]}\big] \right]&=\left(a^*_i a^{\phantom{*}}_k + a^{\phantom{*}}_i a^*_k\right)^{[K]} .\label{eq:Mcavity_algebraK}
	\end{align}
	The relation $\big(a^*_i a^{\phantom{*}}_i \big)^{[K]} \in \LieC{\id^{[K]}, (\sigma_i^z)^{[K]} } $was shown to hold for $M=2$ in Lemma~\ref{lemma_technical_2}. It can be easily seen that this generalizes from Eq.~\eqref{eq:technical_lemma_a*a} where $M=2$ to an arbitrary number of cavities $M\ge2$. Combining this with Eq.~\eqref{eq:Mcavity_algebraK} gives
	\begin{equation}\label{eq:suff_hop_proof2}
	\left(a^*_i a^{\phantom{*}}_k + a^{\phantom{*}}_i a^*_k\right)^{[K]} \in \LieC{\HH{I}^{[K]}, \id^{[K]}, (\sigma_i^z)^{[K]}, (\sigma_k^z)^{[K]} } .
	\end{equation}
	Since the operators considered commute with $N$, they generate a Lie subalgebra of $\mathfrak{u}(N)$ (c.f Propositions 4.4 and 4.5 from \cite{keyl2014controlling}). Therefore, we can apply the same argument as in Proposition~\ref{prop:suff_trunc_compl} and generalize from the truncated complexified algebras to the corresponding dynamical Lie group: Hence for all $(i,k)\in J$ we find $\exp [it\,( a^*_i a^{\phantom{*}}_k +a^{\phantom{*}}_i a^*_k) ] \in \G (\id,\sigma_i^z,\sigma_k^z, \HH{I}) $.\\
	We can multiply all such exponentials for $(i,k)\in J$ and find
	\begin{equation}
	\exp\Big[it \sum_{(i,k)\in J} (a^*_i a^{\phantom{*}}_k +a^{\phantom{*}}_i a^*_k )\Big]\in \G (\id, \sigma_1^z,\dots, \sigma_{M+1}^z, \HH{I} ).
	\end{equation} This reads as follows: Collective control on $\HH{I}$ is sufficient to approximately tune individual hopping interactions $ a^*_i a^{\phantom{*}}_k + a^{\phantom{*}}_i a^*_k$. This implies the two statements we want to show:
	As first consequence we have for $J_M:= J\setminus \{(M,M+1)\}$ that
	\begin{equation}
	\exp(itH'_\tH ),\, \exp(itH''_\tH)\in \G (\id, \sigma_1^z,\dots, \sigma_{M+1}^z, \HH{I} )	.
	\end{equation} Secondly, if a quantum control system is strongly operator controllable for some graph $J$ then the system corresponding to a supergraph $I$ for which $J$ is a spanning tree is also strongly operator controllable. 
\end{proof}

\begin{lem}[Drift Hamiltonians]\label{lem:suff_drift_JCH}
	As in Lemma~\ref{lem:suff_hop_JCH}, consider a quantum control systems from Proposition~\ref{prop:JCH_con_sys} with $M+1$ cavities. The operators $\DriftJ':=\sum_{k=1}^{M} H_{\JC,k}$ and $\DriftJ'':=\sum_{k=M}^{M+1} H_{\JC,k}$ satisfy
	\begin{equation}
	\exp \left( it\DriftJ'\right),\;\exp \left( it\DriftJ'' \right) \in \G (\id,\sigma_1^z,\dots, \sigma^z_{M+1},\DriftJ)\; .
	\end{equation}
\end{lem}
\begin{proof}
	Let us consider the drift operators $\DriftJ'$, $ \DriftJ''$ and $\DriftJ:=\sum_{k=1}^{M+1} H_{\JC,k}$. For every $K\in \mathbb{N}_0$ and $k\in \{1,\dots,M+1\}$ we find 
	\begin{equation}
	\left[(\sigma_k^z)^{[K]} ,\big[ (\sigma_k^z)^{[K]}, \DriftJ^{[K]} \big]\right]=\oIk (a^*_k \sigma^-_k + a^{\phantom{*}}_k \sigma^+_k)^{[K]} .	
	\end{equation}
	Together with $(a^*_k a^{\phantom{*}}_k )^{[K]}\in \LieC{ \id^{[K]}, (\sigma_k^z)^{[K]} }$ this shows that 
	\begin{equation}
	\left( \DriftJ'\right)^{[K]},\, \left(\DriftJ''\right)^{[K]}\in \LieC{ \id^{[K]}, (\sigma_1^z)^{[K]} ,\dots, (\sigma_{M+1}^z)^{[K]}, \DriftJ^{[K]} }.
	\end{equation} 
	On the group level, we get the desired inclusions. 	
\end{proof}

\begin{figure}[ht]
	\centering
	\def\svgwidth{10cm}
	\large
	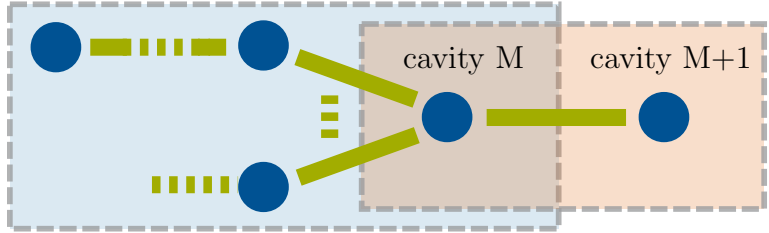
	\caption{Spanning tree graph for a JCH model with $M+1$ cavities. Vertices in blue represent cavities containing a JC model, edges in green represent hopping of photons between cavities. Illustration of the proof of Theorem~\ref{thm:contr_JCH} where the controllability of the $M+1$ cavity system is traced back to the controllability of two overlapping systems: One for $M$ cavities (blue box) and one for 2 cavities (orange box).}
	\label{fig:graph_ind}
\end{figure}

\begin{proof}[Proof of Theorem~\ref{thm:contr_JCH}]
	We show this by induction on the number $M\ge2$ of cavities. The statement for $M=2$ is proven in Theorem~\ref{thm:control_2}.
	
	For the \textbf{inductive step}, assume that the control system from Proposition~\ref{prop:JCH_con_sys} with $M$ cavities is strongly operator controllable and consider the model for $(M+1)$ cavities. As described above, its hopping interaction graph can be given by a set $I\subseteq\{(i,j)|\, 1\le i<j\le M+1\}$. The control system satisfies the recurrence condition due to Theorem~\ref{thm:rec_JCH}. Hence it is strongly operator controllable if and only if
	\begin{equation}\label{eq:proof_Mcavitycontrol_1}
		\G \left(\id,\sigma^z_1, \dots, \sigma^z_{M+1} ,\sigma^x_1, \dots, \sigma^x_{M+1},\DriftJ, \HH{I} \right)=\mathcal{U} \left(\Hil_{M+1} \right) 
	\end{equation} 
	where $\DriftJ:=\sum_{k=1}^{M+1} \left( \oPk \,a^*_ka^{\phantom{*}}_k +\oAk \, \sigma^z_k + \oIk  (a^*_k \sigma^-_k + a^{\phantom{*}}_k \sigma^+_k )\right)$ and $\HH{I}$ is as in Eq.~\eqref{eq:def_HHI} given by $\HH{I}=\sum_{(i,k)\in I} (a^*_i a^{\phantom{*}}_k + a^{\phantom{*}}_i a^*_k)$.
	
	Let $J$ be a spanning tree of $I$. Lemma~\ref{lem:suff_hop_JCH} shows that controllability of a system with this hopping interaction graph $J$ is sufficient for the controllability of the original system $I$. In a connected tree-like graph, one can always find a vertex with the following two properties: (i) it has only one edge connecting it to the rest of the graph and (ii) if one removes it (together with the corresponding edge), then the remaining graph is still a tree-like connected graph. 
	This (and potentially a relabeling of the cavities) allows for assuming that $\{(M,M+1)\}\in J$ and that $(i,M+1)\notin J$ for all $i\neq M$.
	The set $J_M:=J\setminus \{(M,M +1)\}$ represents a connected tree-like graph for $M$ cavities and by induction hypothesis, the corresponding quantum control system is strongly operator controllable. Thus 
	\begin{equation}
		\G \left(\id,\sigma^z_1, \dots, \sigma^z_M ,\sigma^x_1, \dots, \sigma^x_M, \DriftJ', H'_\tH \right)=\mathcal{U}\left(\Hil_{M} \right) 
	\end{equation}
	where $\DriftJ':=\sum_{k=1}^{M} \left( \oPk \,a^*_ka^{\phantom{*}}_k +\oAk \, \sigma^z_k + \oIk  (a^*_k \sigma^-_k + a^{\phantom{*}}_k \sigma^+_k )\right)$ and $H'_\tH:=\sum_{(i,k)\in J_M} (a^*_i a^{\phantom{*}}_k + a^{\phantom{*}}_i a^*_k)$ 
	represent an $M$-cavity drift and hopping Hamiltonian, respectively.
	The subgraph $\{(M,M+1)\}\subseteq J$, that consists of the cavities $M$, $M+1$ and the edge connecting them, is pictured in an orange box in Fig.~\ref{fig:graph_ind}. 
	The corresponding quantum control system is just a two cavity JCH model. It is strongly operator controllable due to Theorem~\ref{thm:control_2}, i.e., we have 
	\begin{equation}
		\G \left(\id,\sigma^z_M, \sigma^z_{M+1} ,\sigma^x_M, \sigma^x_{M+1},H''_0,H''_\tH \right) = \mathcal{U}\left(\Hil_{2} \right)
	\end{equation}
	where $\DriftJ'':=\sum_{k=M}^{M+1} \left( \oPk \,a^*_ka^{\phantom{*}}_k +\oAk \, \sigma^z_k + \oIk  (a^*_k \sigma^-_k + a^{\phantom{*}}_k \sigma^+_k )\right)$ and $H''_\tH:=a^*_M a^{\phantom{*}}_{M+1} + a^{\phantom{*}}_M a^*_{M+1}$. \\
	Hence we constructed two strongly operator controllable quantum control systems whose Hilbert spaces overlap on cavity $M$. We can apply Proposition~\ref{prop:overlapping_Keyl} with the identifications $\mathcal{K}_1=\Hil_{M-1}$ and $\mathcal{K}_2=\Hil_1=\mathcal{K}_3$ to obtain 
	\begin{align}
		\G \left(\id,\sigma^z_1, \dots, \sigma^z_{M+1} ,\sigma^x_1, \dots, \sigma^x_{M+1},\DriftJ',\DriftJ'' ,H'_\tH, H''_\tH \right)=\mathcal{U}\left(\Hil_{M+1} \right) \label{eq:G_JCH_M+1_prime} 
	\end{align}
	Note that the dynamical group on the left hand side is not the one in Eq.~\eqref{eq:proof_Mcavitycontrol_1}. 
	But Lemmas~\ref{lem:suff_hop_JCH} and \ref{lem:suff_drift_JCH} imply that
	\begin{align*}
		\G \left(\id,\sigma^z_1, \dots, \sigma^z_{M+1} ,\sigma^x_1, \dots, \sigma^x_{M+1}, \DriftJ', \DriftJ'',H'_\tH , H''_\tH \right)&\\
		\subseteq \G(\id,\sigma^z_1, \dots, \sigma^z_{M} ,\sigma^x_1, \dots, \sigma^x_{M},\DriftJ, \HH{I} ) &\; .
	\end{align*}
	The dynamical group on the right hand side is by definition a subgroup of the unitary group which is due to Eq.~\eqref{eq:G_JCH_M+1_prime} equal to the left hand side. We conclude that Eq.~\eqref{eq:proof_Mcavitycontrol_1} holds and hence that the system is strongly operator controllable for $M+1$ cavities.
\end{proof}

\subsection{Discussion of results}
We proved the main result of the article (Theorem~\ref{thm:contr_JCH}) stating a control system for the JCH model which is strongly operator controllable. We now discuss the assumptions that we made for this statement. 

The system's free and uncontrolled evolution is governed by the drift Hamiltonian 
\begin{equation}
	\DriftJ= \sum_{i=1}^M \left[ \oP{i} a^*_ia^{\phantom{*}}_i +\oA{i} \sigma^z_i+ \oI{i} (a^*_i \sigma^-_i + a^{\phantom{*}}_i \sigma^+_i )\right]
\end{equation} which corresponds to the sum of JC Hamiltonians on every cavity. The crucial element for the controllability proof is the atom-photon interaction in every cavity. Thereby, we indirectly assumed that the corresponding coupling strength $\oI{k}$ is large enough for all cavities $k$ so that its inverse is small compared to the lifetime of the cavities. If the converse was true, i.e., if $(\oI{k})^{-1}$ was large compared to the cavity lifetime then the drift would be negligible. In that case, the control Hamiltonians from Proposition~\ref{prop:JCH_con_sys} would not be sufficient. We would have to include the atom-photon interaction ($\HI$ from Eq.~\eqref{eq:notation_HP_HI}) as an additional control Hamiltonian to achieve strong operator controllability.
Among the control Hamiltonians we find the Pauli matrices $\sigma^x_i$ and $\sigma^z_i$ on all atoms $i=1\dots,M$. By this we assume to be able to individually address and fully control all atoms. Via the hopping control Hamiltonian $\HH{I}$ we can collective turn on (or off) all hopping interactions. In an optical lattice representing the cavities this could be realized via rapidly ramping up the lattice depth. We note that individual addressability (of photons or single hopping interactions) is not necessary to achieve controllability.

We can easily extend the controllability result to certain related models. The first example was given in the previous paragraph where we replaced the drift Hamiltonian by the atom-photon interaction $\HI$ as an additional control Hamiltonian. We discussed in which scenarios this is reasonable.
A second possibility, but from our viewpoint weaker result is the following: we replace the control Hamiltonians $\sigma^z_i$ by  $a^*_ia^{\phantom{*}}_i$. Then the controllability proof becomes even simpler: in Lemma~\ref{lem:suff_hop_JCH} we can stop after Eq.~\eqref{eq:Mcavity_algebraK} and do not need to consider Eq.~\eqref{eq:suff_hop_proof2} (in Lemma~\ref{lemma_technical_2} one easily obtains that $A_1-A_2$ and $A_1+A_2$ are generated by the drift and $a^*_ia^{\phantom{*}}_i$; and we do not need Eq.~\eqref{eq:technical_lemma_a*a} but can stop the proof after Eq.~\eqref{eq:JCH_2_proofAi_1}). Note that this replacement corresponds to assuming individual addressability of all cavity modes and is therefore less practical: we would still have to control all atoms (at least along one axis) via the complementary operators $\sigma^x_i$.
A third related model where we can apply our controllability result to is the JC model. It was shown in~\cite{keyl2014controlling} that the control Hamiltonians $\sigma^x$, $\sigma^z$ and $a\otimes \sigma^+ +a^\dagger\otimes \sigma^-$ achieve strong operator controllability on the Hilbert space $L^2(\mathbb{R})\otimes \mathbb{C}^2$. Here we can generalize this result in the following sense: Assuming that the summand $a\otimes \sigma^+ +a^\dagger\otimes \sigma^-$ appears in the drift Hamiltonian, control of the atom via $\sigma^x$ and $\sigma^z$ is sufficient.
Note that the controllability proof included an induction on the number of cavities. This implies that we cannot say anything about the controllability of an infinite JCH model (with an infinite number of cavities). 

We now discuss some indirect conditions in Theorem~\ref{thm:contr_JCH} that are assumed in the definition of a quantum control system (c.f~Definition~\ref{dfn:con_sys}). In Section~\ref{sec:Problem_Results} we argued why they are necessary for the well-definedness of the problem but here we argue that they can be relaxed for the concrete JCH model. We assumed piecewise constant control functions such that only one of them is nonzero at each time. This means that the controls can be instantly switched on from 0 to a finite value (also vice versa) and that only one of them is `on' at a given time. For the JCH model, the above restriction on the control functions is not necessary. We showed that the control Hamiltonians and the drift generate a Lie algebra such that all its elements admit a joint dense domain of essential self-adjointness. Hence the time evolution operator is well-defined for any element of this Lie algebra. We can therefore take control functions that are piecewise continuous, and that are simultaneously switched on.

\section{Conclusion and outlook}\label{sec:conclusion}

We showed strong operator controllability of the Jaynes-Cummings-Hubbard model. We assumed that the system's undisturbed evolution is governed by Jaynes-Cummings Hamiltonians of the cavities. Our main result states that by controlling the individual atoms via Pauli matrices $\sigma^x$ and $\sigma^z$ as well as by globally switching on and off the photon hopping between cavities we can achieve the following two things: We can approximately tune the system from a given initial pure state into any desired pure state with arbitrary high precision; we can also approximate any unitary operator on the system by a suitably tuned time evolution operator. 

Although the main result only shows the existence of suitable control functions, the proofs are partly constructive. In the two cavity case, it is possible to deduce a concrete choice of control functions to link two pure states: the proof suggests to steer the system first into the zero state (no photon and all atoms in the ground state) and then into the final state. Obviously, this strategy is not optimal for many pairs of states. Hence it would be a next step to develop concrete and optimal controllability schemes for the Jaynes-Cummings-Hubbard model.

The proof of our main result relies on a strategy to analyze controllability of an infinite dimensional system by exploiting its symmetries. We build on former results~\cite{keyl2014controlling} and generalize them to control systems with a non-zero drift.  
This constitutes the second result of this article that goes beyond the JCH model. We give a list of sufficient conditions for strong operator controllability and provide a proof in the full infinite dimensional setting analyzing convergence in the strong operator topology. The key condition is that all control Hamiltonians but one admit a U$(1)$ symmetry, i.e., they commute with a symmetry operator $N$ with eigenvalues $n\in \mathbb{N}_0$ of finite multiplicity. Then, the symmetry obeying operators generate a block diagonal Lie algebra that can be written as an infinite direct sum of finite dimensional Lie algebras. The key idea is that one can cut off this decomposition at a sufficiently high number without sacrificing approximations in the strong sense. To analyze controllability one hence has to check inclusions on an increasing sequence of finite dimensional Lie algebras. Since the last control Hamiltonian breaks this symmetry in a special way, it ensures controllability.
In contrast to other methods that make systems effectively finite dimensional by truncating the Hilbert space, we now provide convergence analysis in the strong operator topology on the full Hilbert space. Hence we establish Lie theoretic tools to grasp the full infinite dimensionality of the system.

A direct application of our symmetry methods to other systems would be a natural next step: In such systems, some control Hamiltonians would have to commute with a number operator that conserves e.g.~some charge, excitation or occupation number. Note that the finite multiplicity of its eigenvalues is crucial for every statement that we make. This prohibits a direct use for systems with an infinite number of modes, i.e.~the full electromagnetic fields. 
But nevertheless, these methods might be used as a starting point in this direction. 
Another direction of future research would be to study open quantum control systems and to include cavity losses in the JCH model.

\section*{Acknowledgements}
The authors would like to thank Thomas Schulte-Herbr\"uggen for valuable discussions. MH is supported by the International Max Planck Research School for Quantum Science and Technology at the Max-Planck-Institut f\"ur Quantenoptik.


\end{document}